\documentclass{article}
\usepackage{fullpage}

\usepackage[usenames]{color}
\usepackage[colorlinks = true]{hyperref}
\usepackage[english]{babel}   
\usepackage{bm,bbm,amsmath,amssymb,amsthm,graphicx}
\usepackage{url}
\usepackage[caption=false,font=normalsize,labelfont=sf,textfont=sf]{subfig}
\usepackage{algorithm}
\usepackage{algpseudocode}
\usepackage{tikz}
\usetikzlibrary{calc,patterns,decorations.pathmorphing,decorations.markings}

\newcommand{\change}[1]{\ensuremath{\operatorname{#1}}}
\newcommand{\MAT}{\left[ \begin{array}}
\newcommand{\mat}{\end{array} \right]}

\newtheorem{Corollary}{Corollary}[section]
\newtheorem{Lemma}{Lemma}[section]
\newtheorem{Theorem}{Theorem}[section]

\newtheorem{Q}{Question}[section]

\usepackage{etoolbox}
\makeatletter
\patchcmd{\@makecaption}
  {\scshape}
  {}
  {}
  {}
\makeatother

\def \a {\bm{a}}
\def \A {\mathbf{A}}
\def \AA {\mathcal{A}}

\def \At {\widetilde{\mathbf{A}}}

\def \c {\bm{c}}
\def \cb {\breve{c}}

\def \ct {\widetilde{c}}

\def \CC {\mathcal{C}}
\def \CCC {\mathbb{C}}

\def \D {\mathbf{D}}

\def \e {\bm{e}}
\def \E {\mathbf{E}}
\def \EEE{\mathbb{E}}

\def \fh {\widehat{f}}

\def \fs{f^\star}
\def \FF {\mathcal{F}}

\def \HH {\mathcal{H}}
\def \I {\mathbf{I}}

\def \Kh{\widehat{K}}

\def \LL {\mathcal{L}}

\def \NN {\mathcal{N}}
\def \OO {\mathcal{O}}

\def \PP {\mathcal{P}}

\def \Q {\mathbf{Q}}
\def \QQ {\mathcal{Q}}

\def \rh{\widehat{r}}
\def \rs{r^\star}
\def \R {\mathbf{R}}
\def \RR {\mathcal{R}}
\def \RRR {\mathbb{R}}
\def \S {\mathbf{S}}

\def \T {\mathbf{T}}

\def \u {\bm{u}}

\def \U {\mathbf{U}}

\def \V {\mathbf{V}}

\def \W {\mathbf{W}}

\def \x {\bm{x}}
\def \X {\mathbf{X}}

\def \Xh {\widehat{\mathbf{X}}}
\def \Xs {\mathbf{X}^\star}

\def \y {\bm{y}}
\def \Y {\mathbf{Y}}

\def \Z {\mathbf{Z}}

\def \sumk{\sum_{k=1}^K}

\def \bPhi {\boldsymbol{\Phi}}
\def \bphi {\boldsymbol{\phi}}

\def \bSigma {\boldsymbol{\Sigma}}

\def \zero {\mathbf{0}}

\begin{document}

\title{Recovery Analysis of Damped Spectrally Sparse Signals and Its Relation to MUSIC}

\author{Shuang~Li,
        Hassan Mansour,
        and Michael B. Wakin
\thanks{SL and MBW are with the Department of Electrical Engineering, Colorado School of Mines, Golden, CO 80401. Email: \{shuangli,mwakin\}@mines.edu. HM is with Mitsubishi Electric Research Laboratories, Cambridge, MA 02139. Email: mansour@merl.com.}
}

%

\maketitle

\begin{abstract}
{One of the classical approaches for estimating the frequencies and damping factors in a spectrally sparse signal is the MUltiple SIgnal Classification (MUSIC) algorithm, which exploits the low-rank structure of an autocorrelation matrix. Low-rank matrices have also received considerable attention recently in the context of optimization algorithms with partial observations, and nuclear norm minimization (NNM) has been widely used as a popular heuristic of rank minimization for low-rank matrix recovery problems. On the other hand, it has been shown that NNM can be viewed as a special case of atomic norm minimization (ANM), which has achieved great success in solving line spectrum estimation problems. However, as far as we know, the general ANM (not NNM) considered in many existing works can only handle frequency estimation in undamped sinusoids. In this work, we aim to fill this gap and deal with damped spectrally sparse signal recovery problems. In particular, inspired by the dual analysis used in ANM, we offer a novel optimization-based perspective on the classical MUSIC algorithm and propose an algorithm for spectral estimation that involves searching for the peaks of the dual polynomial corresponding to a certain NNM problem, and we show that this algorithm is in fact equivalent to MUSIC itself. Building on this connection, we also extend the classical MUSIC algorithm to the missing data case.  We provide exact recovery guarantees for our proposed algorithms and quantify how the sample complexity depends on the true spectral parameters. In particular, we provide a parameter-specific recovery bound for low-rank matrix recovery of jointly sparse signals rather than use certain incoherence properties as in existing literature. Simulation results also indicate that the proposed algorithms significantly outperform some relevant existing methods (e.g., ANM) in frequency estimation of damped exponentials.}
{Spectral estimation, nuclear norm minimization, atomic norm minimization, MUSIC algorithm, low-rank matrix completion.}
\end{abstract}


\section{Introduction}
\label{intr}

In this paper, we consider the problem of identifying the frequencies and damping factors contained in a spectrally sparse signal, namely, a superposition of a few complex sinusoids with damping, either from a complete set of uniform samples (which we refer to as full observations) or from a random set of partial uniform samples (which we refer to as the missing data case). This kind of signal arises in many applications, such as nuclear magnetic resonance spectroscopy~\cite{umesh1996estimation, qu2015accelerated}, radar processing~\cite{xie2017radar, zhu2016dimensionality}, modal analysis~\cite{park2015sampling, li2017atomicIEEE}, and electric motor fault detection~\cite{xie2020robust}.
It is well known that the frequencies and damping factors can be identified by the classical spectrum estimation approaches, such as Prony's method~\cite{de1795essai}, the Matrix Pencil method~\cite{hua1990matrix}, and the MUltiple SIgnal Classification (MUSIC) algorithm~\cite{schmidt1981signal, schmidt1986multiple}, when full observations are available. However, in many real-world applications, obtaining such full observations with high speed uniform sampling is of high cost and technically prohibitive. Lower-rate, nonuniform sampling can be an appealing alternative~\cite{marvasti2012nonuniform,li2017atomicIEEE,Li15,yang2014exact,park2014modal} and results in the partial observations (missing data) discussed in this work. 

The MUSIC algorithm, which is widely used in signal processing~\cite{thomson1982spectrum, krim1996two}, was first proposed by Schmidt as an improvement to Pisarenko's method~\cite{schmidt1981signal}. MUSIC exploits the low-rank structure of an autocorrelation matrix, which is divided into the noise subspace and signal subspace via an eigenvalue decomposition. The spectral parameters are then identified by searching for the zeros of a noise-space correlation function~\cite{liao2016music}. The MUSIC algorithm can be used either for spectral analysis of one signal (the single measurement vector, or SMV, problem) or for  multiple measurement vector (MMV) problems involving joint sparse frequency estimation~\cite{yang2014exact}. However, a limitation of these classical spectral estimation methods is that they are not compatible with the  random sampling or compression protocols that can be used to reduce the front-end sampling burden. One recent work~\cite{suryaprakash2012performance} does adapt the MUSIC algorithm to the setting with noisy missing data, the authors provide asymptotic theoretical guarantees on the performance of a singular value decomposition (SVD) on the noisy partially observed data matrix. In contrast, in this work we consider two settings---the (noiseless and noisy) full observation case and the noiseless missing data case---and establish non-asymptotic theoretical guarantees for our proposed algorithms.

We focus on both the SMV and MMV settings in this paper. Samples of the spectrally sparse vector-valued signal (MMV setting) considered in this work can be arranged into a low-rank matrix while samples of the spectrally sparse scalar-valued signal (SMV setting) can be used to form a Hankel matrix, which is also a low-rank matrix. Low-rank matrices have received considerable attention recently in the context of optimization algorithms with partial observations. In particular, low-rank matrix recovery from missing data appears in many practical problems such as matrix completion~\cite{candes2009exact,cai2010singular}, low-rank approximation~\cite{gillard2013optimization, larsson2016convex}, system identification~\cite{fazel2001rank, fazel2013hankel}, and image denoising~\cite{gu2014weighted,zhang2014hyperspectral}. A common approach for recovering a low-rank matrix is known as rank minimization. However, rank minimization problems are, in general, NP-hard. Fortunately, a popular heuristic of rank minimization problems, nuclear norm minimization (NNM), performs very well in low-rank matrix recovery when certain conditions on the measurement system are satisfied~\cite{candes2009exact}.
Recently, it has been shown that NNM for low-rank matrix recovery can be viewed as a special case of atomic norm minimization (ANM) when the atoms are composed of rank one matrices~\cite{candes2014towards, chandrasekaran2012convex}. ANM is a general optimization framework for decomposing structured  signals and matrices into sparse combinations of continuously-parameterized atoms from some dictionary, and one of the primary successes of ANM has been in solving the line spectrum estimation problem in both the complete and missing data cases. Most of the theory for ANM in line spectrum estimation has relied on insight gained from analyzing the dual solution to the ANM problem. \textcolor{black}{However, as far as we know, the general ANM (not NNM) formulation considered in many existing works can only handle frequency estimation in undamped sinusoids~\cite{candes2014towards, tang2013compressed, yang2014exact, Li15, li2019atomic}.} This is due to the existence of an SDP form for ANM when there is no damping contained in the signals. In this work, we aim to fill this gap and identify both the frequencies and damping factors contained in a spectrally sparse signal.

The fact that NNM is a special case of ANM suggests that ANM-type dual analysis can {\em also} be used for NNM. In particular, in this paper, we propose an algorithm for spectral estimation that involves searching for the peaks of the dual polynomial corresponding to the NNM problem. We name this algorithm NN-MUSIC (nuclear norm minimization view of MUSIC), and we highlight the fact that in the full observation case, NN-MUSIC is in fact {\em equivalent} to MUSIC itself. We also provide one such development in this paper: unlike classical MUSIC, the NN-MUSIC algorithm can be naturally generalized to the missing data case, and so we also propose and analyze such a Missing Data MUSIC (MD-MUSIC) algorithm in this paper. \textcolor{black}{MD-MUSIC is not equivalent to first using NNM to complete the missing data and second running conventional MUSIC on the full data matrix; rather, it involves extracting frequency estimates directly from the NNM dual polynomial, and we demonstrate that it can succeed even when ``two-step'' algorithms fail.} Both NN-MUSIC and MD-MUSIC can deal with damped sinusoids. \textcolor{black}{Our simulations also illustrate the advantage of these two proposed algorithms over ANM in frequency estimation of damped sinusoids.}

Using our analytical framework, we also provide exact recovery guarantees for both NN-MUSIC and MD-MUSIC. For NN-MUSIC, our theorem indicates that we can perfectly identify the spectral parameters by searching for the locations in the damping-frequency plane where the $\ell_2$-norm of the dual polynomial achieves $1$, as long as the true spectral parameters are distinct from each other and the number of uniform samples is larger than the number of spectral parameters. For MD-MUSIC, our theory shows that we can perfectly identify the spectral parameters with high probability by searching for the locations in the damping-frequency plane where the $\ell_2$-norm of the dual polynomial achieves $1$ if the number of random samples is sufficiently large, the true spectral parameters are distinct from each other, and the number of uniform samples (from which the random samples are drawn) is larger than the number of spectral parameters. \textcolor{black}{It is even possible to get perfect parameter recovery without exact data recovery.} Moreover, \textcolor{black}{we provide a parameter-specific recovery bound for low-rank matrix recovery of jointly sparse signals, that is, we quantify how the sample complexity depends on the true spectral parameters rather than use certain incoherence properties as in existing literature}.

The remainder of this paper is organized as follows. In Section~\ref{prob}, we introduce both the SMV and MMV settings considered in this paper. In Section~\ref{prior}, we review the classical MUSIC algorithm as well as its variants. In Section~\ref{main}, we offer a novel optimization-based perspective on the MUSIC algorithm by highlighting the fact that the proposed NN-MUSIC algorithm is equivalent to MUSIC in the full observation case. We also generalize it to the missing data case and propose the MD-MUSIC algorithm to support the idea that this connection between NNM and MUSIC could lead to future developments and understanding. The proofs for theoretical guarantees are presented in Section~\ref{proof}. In Section~\ref{nume}, we explore the recovery performance of the proposed NN-MUSIC and MD-MUSIC algorithms with numerical simulations. Finally, we conclude this work and discuss future directions in Section~\ref{conc}.

\section{Signal Models}
\label{prob}

We are interested in identifying the frequencies and damping factors contained in a spectrally sparse signal, which can be a scalar-valued signal in the SMV setting or a vector-valued signal in the MMV setting. We first introduce the SMV and MMV settings that are considered in this work. Throughout this work, we use superscript ``$\circ$'' to denote row vectors, and superscripts $``\top"$ and $``H"$ to denote transpose and conjugate transpose, respectively.

\subsection{Single Measurement Vector (SMV) setting}
\label{smv}

In the SMV setting, a scalar-valued, continuous-time signal is assumed to have the form

\begin{align}
y(t) = x(t) + e(t), ~~x(t) = \sumk c_k r_k^t  e^{j2\pi f_k t}, \label{smvsignal}
\end{align}
where $\{c_k\}$, $\{r_k\}$, $\{f_k\}$ and $e(t)$ are the unknown \textcolor{black}{complex} coefficients, damping ratios, frequency parameters, and additive observation noise, respectively. Such signals appear in many applications, such as radar, sonar, and communications. Without loss of generality, we assume the frequencies $\{f_k\}$ belong to the interval $[0,1)$, the damping ratios $\{r_k\}$ belong to the interval $[0,1]$, the \textcolor{black}{complex} coefficients \textcolor{black}{$c_k\neq 0$}, and $e(t)\sim \CC\NN(0,\sigma^2)$.

\subsection{Multiple Measurement Vector (MMV) setting}
\label{mmv}

In the MMV setting, we consider a vector-valued signal $\y^\circ(t) \in \CCC^{1\times N}$, which is a superposition of $K$ damped sinusoids with additive observation noise  $\e^\circ(t) \in \CCC^{1\times N}$. More precisely,
\begin{align}
\y^\circ(t) = \x^\circ(t) + \e^\circ(t),~~ \x^\circ(t)= \sum_{k=1}^K c_k r_k^t e^{j2\pi f_k t} \bphi_k^\top, \label{signal}
\end{align}
with \textcolor{black}{$c_k\neq0$},
$f_k\in[0,1)$ and $r_k\in[0,1]$ being the $k$-th \textcolor{black}{complex} coefficient,  frequency, and damping factor, respectively. Here, each $\bphi_k\in \CCC^{N}$ is a normalized vector ($\|\bphi_k\|_2=1$) that can be viewed as the mode shape in modal analysis problems~\cite{park2015sampling, li2017atomicIEEE}.

%

Suppose we take $M$ uniform samples and arrange $\y^{\circ}_m = \y^\circ(m)$ as the $m$-th row of a data matrix $\Y\in\CCC^{M\times N}$. Define $\Xs \triangleq \left[(\x_0^\circ)^\top~(\x_1^\circ)^\top~\cdots~(\x_{M-1}^\circ)^\top
\right]^\top$ and $\E \triangleq \left[(\e_0^\circ)^\top~(\e_1^\circ)^\top~\cdots~(\e_{M-1}^\circ)^\top
\right]^\top$ as the noiseless data matrix and observation noise matrix, respectively.
Then, we have
\begin{equation}
\begin{aligned}
\Y &= \left[(\y_0^\circ)^\top~(\y_1^\circ)^\top~\cdots~(\y_{M-1}^\circ)^\top
\right]^\top =\sum_{k=1}^K \ct_k \a(r_k,f_k) \bphi_k^\top +\E 
= \A_{rf} \D_{\ct} \bPhi^\top +\E  = \Xs + \E
\label{data}
\end{aligned}
\end{equation}
with\footnote{Note that we abbreviate $\a(r,f)$ to $\a(f)$ when $r=1$.}
\begin{align}
\!\!\!\!\a(r,\!f) \!\triangleq \!\!
\begin{cases}
\!\!\sqrt{\frac{1-r^2}{1-r^{2M}}} [1~\!r e^{j2\pi f 1}\! \cdots r^{M\!-\!1}\!e^{j 2\pi f (M-1)}]^\top\!\!,\!\!\!\!&r\!<\!1,\\
\!\!\frac{1}{\sqrt{M}} [1~ e^{j2\pi f 1} ~\cdots~e^{j 2\pi f (M-1)}]^\top,&r\!=\!1,
\end{cases} \label{defatom}
\end{align}
and
\begin{align*}
\ct_k \triangleq
\begin{cases}
c_k \sqrt{\frac{1-r_k^{2M}}{1-r_k^2}},~&r_k<1,\\
c_k \sqrt{M},~&r_k=1,
\end{cases}
~~~~~~k=1,\ldots,K.
\end{align*}
In addition, we define
$\A_{rf} \triangleq [\a(r_1,f_1),\cdots,\a(r_K,f_K)]$, $\D_{\ct} \triangleq \change{diag}([\ct_1,\cdots,\ct_K])$, and $\bPhi \triangleq [\bphi_1,\cdots,\bphi_K]$.

Let $\y_n$, $\x_n$ and $\e_n$ denote the $n$-th column of $\Y$, $\Xs$ and $\E$, respectively. It can be seen that
\begin{align}
\y_n \!=\! \x_n+\e_n \!=\!  \sum_{k=1}^K \cb_{n,k} \a(r_k,f_k) + \e_n,~n=1,\ldots,N, \label{mmvsignal}
\end{align}
where $\cb_{n,k} = \ct_k \phi_{n,k}$ with $\phi_{n,k}$ being the $(n,k)$-th entry of $\bPhi$. In this model, the observed data consists of $N$ observed length-$M$ signals, each comprised of $K$ damped sinusoids. The $N$ signals share the same set of unknown frequencies and damping factors, but each has a unique set of coefficients.

\section{Prior Work}
\label{prior}

In this section, we review the classical MUSIC algorithm~\cite{schmidt1981signal, schmidt1986multiple} as well as its two variants, Damped MUSIC (DMUSIC)~\cite{li1999improved} and MUSIC adapted to missing data with Gaussian white noise (denoted as MN-MUSIC)~\cite{suryaprakash2012performance}.

\subsection{MUltiple SIgnal Classification (MUSIC) algorithm}
\label{MUSIC}

\subsubsection{SMV MUSIC via autocorrelation matrix}

By sampling the scalar-valued, continuous-time signal $y(t)$, defined in~(\ref{smvsignal}), at $M$ equally spaced times, one can define a vector $\y(t)\in\CCC^{M}$ as
\begin{align}
\label{eq:vecy}
\y(t) \triangleq \left[ y(t)~y(t+1)~\cdots~y(t+M-1) \right]^\top,
\end{align}
which has the autocorrelation matrix
$
\R_y \triangleq \EEE\{\y(t)\y(t)^H\}.
$

The classical MUSIC algorithm aims to identify the unknown frequencies $\{f_k\}$ by constructing (and then decomposing) an estimate of the autocorrelation matrix $\R_y$ without damping, namely, in the case where all $r_k=1$~\cite{stoica1997introduction}. This requires $M > K$. Specifically, consider a full set of uniform observations $\{y(t)\}$ with $t = 0,1,\ldots,L-1$, for some $L > M$. Then, the following sample autocorrelation matrix can be used to approximate $\R_y$:
\begin{align}
\widehat{\R}_y = \frac{1}{L-M+1} \sum_{t = 0}^{L-M} \y(t) \y(t)^H. \label{sampleauto}
\end{align}
Let $[\widehat{\U}_s~\widehat{\U}_n]$ denote the orthonormal eigenvectors of $\widehat{\R}_y$. In particular, suppose $\widehat{\U}_s \in \CCC^{M \times K}$ (signal space) and $\widehat{\U}_n\in\CCC^{M \times (M-K)}$ (noise space) are associated with the $K$ largest eigenvalues and the $M-K$ smallest eigenvalues of $\widehat{\R}_y$, respectively. Then, we summarize the classical MUSIC algorithm in Algorithm~\ref{alg_MUSIC}.
\begin{algorithm}[H]
\caption{MUSIC}\label{alg_MUSIC}
\begin{algorithmic}[1]
\Procedure{Input}{$\{y(t)\}_{t=0}^{L-1},~K$}
\State compute the autocorrelation matrix $\widehat{\R}_y$ as in~(\ref{sampleauto}) and its eigenvectors $\widehat{\U}_n$
\State compute the pseudospectrum: $1/\|\widehat{\U}_n^H \a(f)\|_2^2$, where $\a(f)$ is defined in~\eqref{defatom} with $r=1$
\State localize the $K$
largest local maxima of pseudospectrum to get $\widehat{f}_k$ 
\State \textbf{return} $\widehat{f}_k$
\EndProcedure
\end{algorithmic}
\end{algorithm}

The intuition behind the MUSIC algorithm comes from the fact that, as a consequence of the scalar-valued signal model in~\eqref{smvsignal}, the vector-valued signal $\y(t)$ in~\eqref{eq:vecy} can be expressed as
\begin{align*}
\y(t) =  \sumk \sqrt{M} c_k e^{j2\pi f_k t} \a(f_k) +\bm{e}(t)
=\A_f \c(t) + \e(t)
\end{align*}
with
\begin{align*}
\A_f &= [\a(f_1),~\a(f_2),~\cdots, \a(f_K)],\\
\c(t) &=\sqrt{M}[c_1 e^{j2\pi f_1 t},~c_2 e^{j2\pi f_2 t},~\cdots,c_K e^{j2\pi f_K t}]^\top,\\
\e(t) &= [e(t),~e(t+1),~\cdots,e(t+M-1)]^\top,
\end{align*}
where $\a(f)$ is defined in~\eqref{defatom} with $r=1$.
Then, the autocorrelation matrix becomes
\begin{align*}
\R_y = \EEE\{\y(t)\y(t)^H\} = \A_f \R_c\A_f^H+\sigma^2\I_M
\end{align*}
if $\c(t)$ is uncorrelated with $\e(t)$.
Here, $\R_c \triangleq  \EEE\{\c(t)\c(t)^H\} $ is the autocorrelation matrix of $\c(t)$ and $\I_M$ denotes the $M\times M$ identity matrix.
Note that the coefficients $\{c_k\}$ may be uncorrelated ($\R_c$ is diagonal) or may contain completely correlated pairs ($\R_c$ is singular).
We are interested in the first case, namely, $\R_c$ is diagonal and positive definite \textcolor{black}{since $c_k\neq0$}.\footnote{As is stated in~\cite{schmidt1986multiple}, in general, $\R_c$ will be ``merely" positive definite to reflect the arbitrary degrees of pair-wise correlations occurring between the coefficients.}
On the other hand, the rank of $\A_f$ is $K$ when all the frequencies $\{f_k\}$ are distinct and $M\geq K$. It follows that the rank of $\A_f \R_c\A_f^H$ is $K$. Let $\{\lambda_m\},~m=1,\ldots,M$ denote the non-increasing eigenvalues of $\A_f \R_c\A_f^H$. Then, we have
$
\lambda_{K+1} = \cdots = \lambda_M = 0.
$
As a consequence, the determinant of $\A_f \R_c\A_f^H$ is
$
\change{det}(\A_f \R_c\A_f^H) = \change{det}(\R_y-\sigma^2 \I_M) = 0,
$
which implies that
$
\lambda^y_m= \sigma^2,~m=K+1,\ldots,M,
$
where $\lambda^y_m$ is the $m$-th non-increasing eigenvalue of $\R_y$. Denoting $\u_m$ as the $m$-th eigenvector of $\R_y$ corresponding to eigenvalue $\lambda^y_m$, we have
\begin{align}
\R_y \u_m = \lambda^y_m \u_m,~m=1,\ldots,M. \label{eigRy}
\end{align}
Replacing $\R_y = \A_f \R_c\A_f^H+\sigma^2\I_M$ into the above equation~(\ref{eigRy}), we have
\begin{align*}
\A_f \R_c\A_f^H \u_m= (\lambda^y_m - \sigma^2) \u_m = \zero~\change{or}~\A_f^H \u_m=\zero
\end{align*}
when $\lambda^y_m = \sigma^2$, or equivalently, $m=K+1,\ldots,M$. Then, $\a(f)$, which is defined in~\eqref{defatom}, is orthogonal to $\u_m,~m=K+1,\ldots,M$ (columns of $\widehat{\U}_n$), when $f=f_k,~k=1,\ldots,K$. Therefore, we can identify the frequencies by localizing the $K$ peaks of the pseudospectrum $1/\|\widehat{\U}_n^H \a(f)\|_2^2$.

\subsubsection{SMV MUSIC via Hankel matrix}

As an alternative to the above autocorrelation matrix, a certain Hankel matrix can also be used in the MUSIC algorithm~\cite{liao2016music}.\footnote{Indeed, Hankel structure has been widely used in a variety of algorithms for spectral estimation in the literature~\cite{yang1996rank, scharf2001toeplitz, andersson2014new, andersson2019fixed}.} In particular, from the same full set of uniform observations $\{y(t)\}$ with $t = 0,1,\ldots,L-1$, one can formulate the Hankel matrix
\begin{align}
\HH_y = \left[
\begin{array}{cccc}
y(0) & y(1) & \cdots& y(N-1)\\
y(1) & y(2) & \cdots& y(N)\\
\vdots&\vdots& &\vdots\\
y(M-1) & y(M) & \cdots& y(L-1)\\
\end{array}
\right] \label{hankelmusic}
\end{align}
for some positive integers $M$ and $N$ satisfying $M+N=L+1$.
Then define the noise-space correlation function $R(f)$ and imaging function $J(f)$ as
\begin{align*}
R(f) = \|\U_n^H \a(f)\|_2, \quad J(f) = \frac{1}{\|\U_n^H \a(f)\|_2}
\end{align*}
with
$
\a(f)=\frac{1}{\sqrt{M}}[1~e^{j2\pi f 1} ~\cdots~e^{j 2\pi f (M-1)}]^\top
$
as defined in~\eqref{defatom}.
Here, $\U_n$ spans the noise subspace and contains the left singular vectors of $\HH_y$ corresponding to the $M-K$ smallest singular values.
The frequencies can then be estimated by identifying the $K$ local minima of the noise-space correlation function $R(f)$ or the $K$ local maxima of the imaging function $J(f)$.

Note that the sample autocorrelation matrix in~\eqref{sampleauto} and the Hankel matrix in~\eqref{hankelmusic} are related by
$
\widehat{\R}_y = \frac{1}{L-M+1}\HH_y \HH_y^H.
$
Thus, the eigenvectors of $\widehat{\R}_y$ are the same as the left singular vectors of $\HH_y$ up to a unitary transform. Therefore, the MUSIC algorithm based on the autocorrelation matrix and the Hankel matrix are equivalent since the imaging function $J(f)$ is equivalent to the pseudospectrum in Algorithm~\ref{alg_MUSIC}.

\subsubsection{MMV MUSIC via data matrix}
\label{MUSIC_MMV}

The MUSIC algorithm is also widely used in MMV problems~\cite{liao2014music,iordache2014music, kim2012compressive}.  
\textcolor{black}{Given a multiple measurement matrix $\Y = [\y_1,\cdots,\y_N]$}
(see Section~\ref{MN-MUSIC}), one can directly compute an SVD of $\Y$ to obtain the noise space $\U_n$ from the left singular vectors of $\Y$ and then identify the frequency parameters by localizing the peaks of the imaging function. In particular, denote
$
\Y = [\U_s~\U_n] ~\bSigma~ [\V_s~\V_n]^H
$
as an SVD of the data matrix $\Y$. For the same reason, one can estimate the frequencies by finding the peaks of the imaging function
$
J(f) = \frac{1}{\|\U_n^H \a(f)\|_2}.
$

\subsection{Damped MUSIC (DMUSIC)}
\label{DMUSIC}

In the general model of~\eqref{smvsignal}, the complex-valued sinusoids are damped and decay over time. For this more general case, the DMUSIC algorithm introduced in~\cite{li1999improved} aims to estimate {\em both} the frequencies $\{f_k\}$ and damping ratios $\{r_k\}$ directly using the rank-deficiency and Hankel properties of~(\ref{hankelmusic}). Similar to classical MUSIC, DMUSIC involves constructing the noise subspace matrix $\U_n$ by computing an SVD of the Hankel matrix $\HH_y$. Then, the $(r_k,f_k)$ pairs are identified by finding the peaks of the imaging function
\begin{align}
J(r,f) = \frac{1}{\|\U_n^H \a(r,f)\|_2} \label{Jrf}
\end{align}
with $\a(r,f)$ defined in~\eqref{defatom}.

The intuition behind DMUSIC is that the Hankel matrix in~(\ref{hankelmusic}) can be rewritten as
\begin{align*}
\HH_y = \A_{rf} \D_c (\A_{rf}^N)^\top + \HH_e,
\end{align*}
where $ \HH_e$ is a Hankel matrix formulated with $\{e(t)\},~t=0,\ldots,L-1$, and $\D_c$ is a diagonal matrix with diagonal entries being the scaled coefficients $c_k$. Precisely, the $k$-th diagonal entry of $\D_c$ is $ \frac{c_k \sqrt{(1-r_k^{2M})(1-r_k^{2N})} }{   (1-r_k^2) }$.  $\A_{rf}$ and $\A_{rf}^N$ are Vandermonde matrices defined as
\begin{align*}
\A_{rf} \triangleq  [\a(r_1,f_1),\cdots,\a(r_K,f_K)],\quad
\A_{rf}^N \triangleq  [\a_N(r_1,f_1),\cdots,\a_N(r_K,f_K)],
\end{align*}
with
\begin{align}
\a(r,f) & \!\triangleq\!\! \sqrt{\! \frac{1-r_k^2}{1-r_k^{2M}} }  [1~r e^{j2\pi f 1} \cdots r^{M-1}e^{j 2\pi f (M-1)}]^\top, \nonumber\\
\!\!\!\a_N(r,f) &\!\triangleq\!\!  \sqrt{\! \frac{1-r_k^2}{1-r_k^{2N}} } [1~r e^{j2\pi f 1} \cdots r^{N-1}e^{j 2\pi f (N-1)}]^\top. \label{defatomN}
\end{align}
Note that we add a subscript ``$N$" in~\eqref{defatomN} to distinguish $\a_N(r,f) \in \CCC^{N}$ from $\a(r,f)\in \CCC^{M}$.
When $M,N\geq K$ and all the $(r_k,f_k)$ pairs are distinct, $\A_{rf}$ and $\A_{rf}^N$ are full rank. Then, $\HH_x \triangleq \A_{rf} \D_c (\A_{rf}^N)^\top$ is of rank $K$. Now, consider the case when there is no noise, i.e., $\HH_y = \HH_x$. Denote an SVD of $\HH_y$ as
\begin{align*}
\HH_y = [\U_s~\U_n] \bSigma \left[
\begin{array}{c}
\V_s^H\\
\V_n^H
\end{array}
\right].
\end{align*}
One can show that the range spaces of $\HH_y$, $\A_{rf}$, and $\U_s$ are all equal when there is no noise. Then, $\a(r,f)$ is orthogonal to the columns of $\U_n$ when $(r,f)=(r_k,f_k),~k=1,\ldots,K$. If noise exists, the orthogonal relationship between $\a(r,f)$ and $\U_n$ no longer holds. However, one can identify all the $(r_k,f_k)$ pairs by finding the peaks of the imaging function defined in~(\ref{Jrf}), that is, searching for $\a(r,f)$ that are most nearly orthogonal to the noise space $\U_n$.

\subsection{MN-MUSIC for missing and noisy data}
\label{MN-MUSIC}

The classical MUSIC algorithm has also been adapted to the missing data case with Gaussian white noise (denoted as MN-MUSIC) for applications such as direction of arrival (DOA) estimation~\cite{suryaprakash2012performance}. The authors consider the MMV setting as introduced in Section~\ref{MUSIC_MMV}. More precisely, consider an observed $M \times N$ matrix $\Y = [\y_1,\cdots,\y_N]$, where $\y_n$ is defined in~\eqref{mmvsignal} and repeated as follows
\begin{align*}
\y_n = \sum_{k=1}^K \cb_{n,k} \a(f_k) + \e_n,~~n=1,\ldots,N,
\end{align*}
with $r=1$ since undamped signals are considered in~\cite{suryaprakash2012performance}.

Assume we partially observe the entries of $\Y$ with i.i.d. Bernoulli randomly sampled locations $\Omega \subset \{1, \ldots, M\} \times \{1,\ldots, N \}$. Let $\Y_\Omega$ be the projection matrix of $\Y$ on the index set $\Omega$, i.e
\begin{align*}
{\Y_\Omega}_{ij}=
\begin{cases}
\Y_{ij},~~ &(i,j) \in \Omega,\\
0,~~&else.
\end{cases}
\end{align*}
Then in MN-MUSIC, an SVD is directly performed on $\Y_\Omega$ to get the signal space matrix $\U_s$, which contains the left singular vectors of $\Y_\Omega$ corresponding to the $K$ largest singular values. Finally, the frequencies are estimated by finding the peaks of
$
\|\U_s^H\a(f)\|_2^2,
$
which is essentially same as in Sections~\ref{MUSIC} and~\ref{DMUSIC}.

\section{Main Results}
\label{main}

In this section we outline a connection between MUSIC and low-rank matrix optimization using nuclear norm minimization (NNM), and based on this connection we propose an extension of MUSIC that is appropriate for the missing data case. Our interest in NNM here is specifically due to its connection with MUSIC. There are, of course, alternative low-rank optimization problems that can also be used for spectral analysis. Among these, atomic norm minimization (ANM) has been proposed and analyzed for solving the undamped line spectrum estimation problem in both the full and missing data cases~\cite{yang2014exact, tang2013compressed}. Moreover, a low-rank Hankel matrix recovery problem has recently been considered for damped spectral analysis~\cite{xu2018sep}; that work involves solving the NNM~\eqref{NNM} and~\eqref{NNMmiss} with an extra Hankel constraint on $\X$. While these alternative frameworks have some benefits, we believe that our work sheds light on a more fundamental problem, given the considerable attention that MUSIC has received over the last several decades. This understanding may lead to new developments for MUSIC and other optimization algorithms for spectral analysis in the future.

\subsection{Optimization connection to MUSIC in the full data case}
\label{full}

In this section, we consider both the SMV and MMV settings. Given a set of uniform samples from the signal model~\eqref{smvsignal} in the SMV setting and the data matrix $\Xs = \A_{rf} \D_{\ct} \bPhi^\top$ or its noisy version $\Y$~\eqref{data} in the MMV setting, our goal is to identify the frequencies $\{f_k\}$ and damping factors $\{r_k\}$. Note that in the SMV setting, we can construct a Hankel matrix as in~\eqref{hankelmusic}. As is shown in Section~\ref{DMUSIC}, this Hankel matrix $\HH_x$ can be decomposed as $\HH_x \triangleq \A_{rf} \D_c (\A_{rf}^N)^\top$ and is of rank $K$ when there is no noise. One can observe that both $\Xs$ and $\HH_x$ are low-rank matrices and have the same type of decompositions. Therefore, the analysis on $\Xs$ can also be applied to $\HH_x$, which implies that the algorithms we build using $\Xs$ in the MMV scenario also work for the SMV scenario.

Assume that $\Xs$ is given and $K \ll M,N$, note that $\Xs$ in~(\ref{data}) is low rank. Inspired by the low-rank property of $\Xs$ and the dual analysis that is commonly used in atomic norm minimization (ANM)~\cite{candes2014towards,tang2013compressed}, let us consider the following nuclear norm minimization (NNM)
\begin{align}
\Xh = \arg\min_{\X}~\|\X\|_* ~~\change{s.t.}~\X = \Xs. \label{NNM}
\end{align}
Although this problem has a trivial solution (namely, $\Xh = \Xs$), it is interesting because we can compute the corresponding dual feasible point $\Q$, which is a solution of the dual problem, via the Lagrange function of~(\ref{NNM}) and thus identify the frequencies and damping factors that are contained in the spectrally sparse signal $\x^\circ(t)$ in~(\ref{signal}). In particular, the Lagrange function is given as
\begin{align*}
\LL(\X,\Q) = \|\X\|_* + \langle \Xs-\X,\Q \rangle_{\RRR} =  \|\X\|_* - \langle \X,\Q \rangle_{\RRR}  + \langle \Xs,\Q \rangle_{\RRR},
\end{align*}
with $\Q$ being the dual variable. $\langle \cdot,\cdot \rangle_{\RRR}$ is defined as the real inner product, i.e.,
 \begin{align*}
 \langle \Xs,\Q \rangle_{\RRR} = \text{Re}(\langle \Xs,\Q \rangle) = \text{Re}(\text{Tr}(\Q^H\Xs))
 \end{align*}
 with $\change{Tr}(\cdot)$ denoting the trace of a matrix.
Then, the subgradient of $\LL(\X,\Q)$ with respect to $\X$ is
\begin{align*}
\frac{\partial \LL(\X,\Q)}{\partial \X} = \partial \|\X\|_* - \Q,
\end{align*}
where $\partial \|\X\|_*$ is the subdifferential of the nuclear norm and given as
\begin{align*}
\partial \|\X\|_* &=\partial \|\Xs\|_*  = \{\Z:~\Z = \U_{X^\star}\V_{X^\star}^H+\W,~\U_{X^\star}^H \W = \zero,~\W\V_{X^\star} = \zero,~\|\W\|\leq 1\}
\end{align*}
since $\Xh=\X = \Xs$. Here, we use $\|\W\|$ to denote the spectral norm of the matrix $\W$. Note that $\Xs = \U_{X^\star} \S_{X^\star} \V_{X^\star}^H$ is a truncated SVD of $\Xs$ with $\U_{X^\star}\in \CCC^{M\times K}$, $\S_{X^\star}\in \RRR^{K\times K}$ and $\V_{X^\star}\in \CCC^{N\times K}$.
We can also construct a
$
\Q\in \partial \|\X\|_*
$
by letting $\zero \in \frac{\partial \LL(\Q,\X)}{\partial \X}$ according to the zero-gradient condition in the Karush-Kuhn-Tucker (KKT) conditions~\cite{boyd2004convex}. Finally, we have
\begin{align}
\Q = \U_{X^\star} \V_{X^\star}^H + \W \label{dualsolu}
\end{align}
with $\U_{X^\star}^H \W = \zero,~\W\V_{X^\star} = \zero,~\|\W\|\leq 1$. Note that the dual solution to~\eqref{NNM} is not unique. In particular, one can verify that any $\Q\in\partial \|\Xs\|_*$ is a dual solution since $\langle \Xs,\Q \rangle_{\RRR} = \|\Xs\|_*$.\footnote{When using CVX with the default solver SDPT3 to solve the SDP form of the NNM problem~\eqref{NNM}, we observe that it returns both the primal solution and a minimum norm dual solution, namely, $\Q = \U_{X^\star} \V_{X^\star}^H$, due to the use of a Conjugate Gradient based algorithm~\cite{toh2012implementation,hayami2018convergence}. } 

Given a dual feasible point $\Q= \U_{X^\star} \V_{X^\star}^H \textcolor{black}{+\W}$, we define the dual polynomial as
\begin{align}
\QQ(r,f) \triangleq \Q^H \a(r,f), \label{dualpoly}
\end{align}
which is inspired by the dual analysis in ANM. The following theorem guarantees that we can identify the true $r_k$'s and $f_k$'s by localizing the places where  $\|\QQ(r,f)\|_2$ achieves 1. Moreover, it also indicates that one does not need a separation condition in this full data noiseless setting. (In some previous work on optimization-based spectral estimation~\cite{yang2014exact}, one needs the minimum separation $\Delta_f$, which is defined in Corollary~\ref{THMdualploymiss}, to be on the order of $\frac{1}{M}$ even for the full data noiseless setting.)

\begin{Theorem} \label{THMdualploy}
Let $\mathcal{RF}$ denote the set of the true damping factor and frequency pairs, i.e., $$\mathcal{RF} = \{(r_1,f_1), \cdots, (r_K,f_K)\}.$$
Given the full data matrix $\Xs$ as in~(\ref{data}), compute its truncated SVD $\Xs = \U_{X^\star} \S_{X^\star} \V_{X^\star}^H$. \textcolor{black}{For all $\Q \in \partial\|\Xs\|_*$ with $\|\W\|<1$}, 
the dual polynomial defined in~(\ref{dualpoly}) satisfies
\begin{align*}
\|\QQ(r_k,f_k)\|_2 &= 1,~\forall~(r_k,f_k)\in\mathcal{RF}, \\
\|\QQ(r,f)\|_2 &< 1,~\forall~ (r,f)\notin\mathcal{RF},
\end{align*}
if $M\geq K+1$, all the $(r_k,f_k)$ pairs in $\mathcal{RF} $ are distinct, and $\bPhi \in \CCC^{N\times K}$ is of rank $K$.
\end{Theorem}

The proof of Theorem~\ref{THMdualploy} is given in Section~\ref{proofTHMdualploy}.
\textcolor{black}{Note that for the case when $\|\W\| = 1$, one may get $\|\QQ(r,f)\|_2  = 1$ for some $(r,f)\notin\mathcal{RF}$, i.e., having some false peaks when checking the norm of the dual polynomial. To remove these false peaks, one can solve the following least squares problem. Denote $\{(\rh_k,\fh_k)\}_{k=1}^{\Kh}$ as the damping factor and frequency pairs estimated by localizing the places where $\|\QQ(r,f)\|_2  = 1$. (To localize the places where $\|\QQ(r,f)\|_2 = 1$, one can use the {\em findpeaks} function in Matlab.) Define $\widehat{\A}_{rf} \triangleq  [\a(\rh_1,\fh_1),\cdots,\a(\rh_K,\fh_K)]$. Then the least squares problem 
$$\min_{\bPhi_c} \|\Xs - \widehat{\A}_{rf} \bPhi_c\|_F$$ 
can be used to remove the false peaks. The $(\rh_k,\fh_k)$ pairs corresponding to zero rows of $\bPhi_c$ can be viewed as false estimates.}
  
Based on the above analysis, we propose the following algorithm, named NN-MUSIC (nuclear norm minimization view of MUSIC algorithm), to estimate the damping factors~$\{r_k\}$ and frequencies $\{f_k\}$ of the damped sinusoids from the data matrix $\Xs$. Note that the step with the highest computational cost is the SVD step, and this needs to be performed only once. 
\begin{algorithm}[H]
\caption{NN-MUSIC}\label{alg_NN-MUSIC}
\begin{algorithmic}[1]
\Procedure{Input}{$\Xs \in \CCC^{M \times N}$}
\State compute truncated SVD of $\Xs$: $\Xs = \U_{X^\star} \S_{X^\star} \V_{X^\star}^H$
\State form the dual feasible point: $\Q = \U_{X^\star} \V_{X^\star}^H$
\State form the dual polynomial: $\QQ(r,f) = \Q^H \a(r,f)$
\State localize the places where $\|\!\QQ(r,f)\!\|_2 \!=\! 1$ to get $(\widehat{r}_k,\widehat{f}_k)$
\State \textbf{return} $(\widehat{r}_k,\widehat{f}_k)$
\EndProcedure
\end{algorithmic}
\end{algorithm}

Note that Algorithm~\ref{alg_NN-MUSIC} is essentially equivalent to the MUSIC (in the undamped case) and DMUSIC (in the damped case) algorithms outlined in Section~\ref{prior}. This is due to the fact that $\|\QQ(r,f)\|_2 = \|\U_{X^\star}^H \a(r,f)\|_2$.\footnote{\textcolor{black}{This fact is key to building the connection between MUSIC and nuclear norm minimization.} This is why we construct $\Q = \U_{X^\star} \V_{X^\star}^H$ explicitly in line 3 of Algorithm~\ref{alg_NN-MUSIC}; this corresponds to a choice of $\W=\zero$ in~\eqref{dualsolu}. \textcolor{black}{Note that other constructions of $\Q$ with $\W \neq \zero$ will still work for localizing the $(r,f)$ pairs but will not equal $\|\U_X^H a(r,f)\|_2$.}
}
When there is no noise, the DMUSIC algorithm and its variants characterize the spectral parameters by locating the zeros of a noise-space correlation function or the peaks of the imaging function, and the proposed NN-MUSIC algorithm identifies the spectral parameters by localizing the $(r,f)$ pairs where $\|\QQ(r,f)\|_2$ achieves $1$.
While MUSIC has been classically understood from an algebraic perspective (owing to its closed form), we believe the derivation of NN-MUSIC offers a novel optimization-based perspective on MUSIC that could lead to future developments and understanding.


We also stress that this connection to MUSIC is unique to NNM and does not apply in general to ANM. In particular, the connection arises specifically because the dual feasible point $\Q = \U_{X^\star} \V_{X^\star}^H$ of NNM induces a dual polynomial that satisfies $\|\QQ(r,f)\|_2 = \|\U_{X^\star}^H \a(r,f)\|_2$. On the other hand, the dual feasible point of ANM formulations does not admit the structure $\Q = \U_{X^\star} \V_{X^\star}^H$, in general.



Finally, consider the case when the given data matrix $\Y$ contains some additive white Gaussian noise, i.e.,
$
\Y = \Xs+\E
$
with $\E$ denoting the measurement noise. Then,
we can solve the following nuclear norm denoising program
\begin{align}
\min_{\X}~\frac{1}{2}\|\Y-\X\|_F^2+ \lambda \|\X\|_*, \label{NNMnoise}
\end{align}
where $\lambda$ is a regularization parameter.
As is shown in the simulation, we can estimate the $(r,f)$ pairs by localizing the peaks of the norm of the corresponding dual polynomial.
We leave the robust performance analysis of this framework for future work. \textcolor{black}{In particular, it would be interesting to characterize the parameter estimation performance of program~\eqref{NNMnoise} in terms of the signal-to-noise ratio and the separation of the true frequencies, similar to the analysis in~\cite{li2018approximate}.}

\subsection{Extension to the missing data case}
\label{miss}

Unlike the classical formulation of MUSIC, the optimization-based derivation of NN-MUSIC allows it to be naturally extended to the missing data case. In particular, assume that we partially observe the entries of the full data matrix $\Xs$ in~(\ref{data}) with
uniformly random sampled locations $\Omega \subset \{1, \ldots, M\} \times \{1,\ldots, N \}$. Let $\X_\Omega = \PP_\Omega(\X)$ be the projection matrix of $\X$ on the index set $\Omega$, i.e
\begin{align*}
{\X_\Omega}_{ij}= \PP_{(i,j)}(\X)=
\begin{cases}
\X_{ij},~~ &(i,j) \in \Omega,\\
0,~~&else.
\end{cases}
\end{align*}
Notice that recovering the missing entries of the matrix $\Xs$ reduces to a matrix completion problem~\cite{candes2009exact}, commonly formulated via the following NNM
\begin{align}
\Xh = \arg\min_{\X}~\|\X\|_*  \quad \change{s.t.}~\X_{\Omega} = \Xs_{\Omega},\label{NNMmiss}
\end{align}
which can be solved by the corresponding semi-definite program (SDP)
\begin{equation}
\begin{aligned}
\min_{\X,\T,\textcolor{black}{\D}}~&\frac{1}{2} \change{Tr}(\T)+\frac{1}{2} \change{Tr}(\textcolor{black}{\D}) \quad 
\change{s.t.}~\left[
\begin{array}{cc}
\T&\X\\
\X^H&\textcolor{black}{\D}
\end{array}
\right]\succeq 0, ~\X_{\Omega} = \Xs_{\Omega}.
\label{SDPmiss}
\end{aligned}
\end{equation}

The dual problem of~\eqref{SDPmiss} is given by
\begin{align}
\textcolor{black}{\max_{\Q}~\langle \Q_\Omega,\Xs_\Omega  \rangle_{\RRR} \quad \change{s.t.}~\|\Q_\Omega\|_2 \leq 1.}
\label{dual_prob}
\end{align}
Therefore, we can define the dual polynomial as
$
\QQ(r,f) \triangleq \Q^H \a(r,f),
$
where $\Q$ is the dual solution.
Similar to Theorem~\ref{THMdualploy}, the following theorem guarantees that we can identify the true $r_k$'s and $f_k$'s by localizing the places where  $\|\QQ(r,f)\|_2$ achieves~1.

\textcolor{black}{
\begin{Theorem} \label{THMdataxrfd}
Suppose $\Xs$ is a data matrix of the form~\eqref{data} and all the $(r_k,f_k)$ pairs are distinct. Given the
uniformly partial random observed data matrix $\Xs_{\Omega}$. 
For any data matrix $\Xh$ obtained by solving the NNM problem~\eqref{NNMmiss}, we denote $\U_{\Xh} \S_{\Xh} \V_{\Xh}^H$ as a truncated SVD of $\Xh$. Denote $\Q_\Omega$ as the projection matrix of the dual solution $\Q$ on the index set $\Omega$. Then, $\Q_\Omega$ has the form $\Q_\Omega = \U_{\Xh}\V_{\Xh}^H + \W$ with $\U_{\Xh}^H \W = \zero,~\W\V_{\Xh} = \zero,~\text{and}~\|\W\| \leq 1$ for some matrix $\W$. Here, we consider dual solutions with $\Q = \Q_\Omega$.\footnote{\textcolor{black}{Recall that we observe CVX with the default solver SDPT3 always returns a dual solution with minimal energy, so the entries of $\Q$ outside the index $\Omega$ are always 0, i.e., we have $\Q = \Q_\Omega$.}} If the range space of $\Xh$ contains the range space of $\Xs$ (i.e., $\RR(\Xh) \supseteq \RR(\Xs)$) and $\|\W\|<1$\footnote{\textcolor{black}{Same as in the full data case, when $\|\W\| = 1$, one may get $\|\QQ(r,f)\|_2  = 1$ for some $(r,f)\notin\mathcal{RF}$ with $\a(r,f) \notin \RR(\U_{\Xh})$, i.e., having some false peaks when checking the norm of dual polynomial. Again, one can  remove these false peaks by solving a least squares problem.}}, the dual polynomial $\QQ(r,f) = \Q^H\a(r,f)$ satisfies 
\begin{align*}
\|\QQ(r,f)\|_2 &= 1, \quad\forall~(r,f)\in\mathcal{RF},~\text{or} ~\forall~ (r,f)\notin\mathcal{RF}~\text{with}~\a(r,f) \in \RR(\U_{\Xh}), \\
\|\QQ(r,f)\|_2 &< 1, \quad\forall~ (r,f)\notin\mathcal{RF}~\text{with}~\a(r,f) \notin \RR(\U_{\Xh}).
\end{align*}
\end{Theorem}
}

\textcolor{black}{
The above theorem is proved in Section~\ref{proofTHMdataxrfd}. 
Note that one can even localize the true $(r,f)$ pairs from the dual polynomial when we do not have perfectly recovered data matrix (i.e., $\Xh \neq \Xs$).  
For the case when we do have perfect data recovery, i.e., $\Xh = \Xs$, we can further conclude that $\|\QQ(r,f)\|_2 = 1$ only when $(r,f)\in\mathcal{RF}$ and $\|\QQ(r,f)\|_2 < 1$ as long as $(r,f)\notin\mathcal{RF}$. Moreover, we can also quantify sample complexity needed for perfect data recovery in terms of the explicit parameters such as damping ratios and frequencies instead of some incoherence property. We summarize these results in the following corollary.
}


\begin{Corollary} \label{THMdualploymiss}
Suppose $\Xs$ is a data matrix of the form~\eqref{data} and all the $(r_k,f_k)$ pairs are distinct. Given the
uniformly partial random observed data matrix $\Xs_{\Omega}$, suppose
$
|\Omega| \geq c_1 \mu_1 c_s K \log^4(MN)
$
for some numerical constants $c_1 > 0$ and $c_s \triangleq \max\{M,N\}$. Here,
$
\mu_1\geq
\max \left\{\frac{M} {\LL(M,r,f)}, \frac{\mu_2}{\sigma_{\min}^2(\bPhi)} \right\}
$
denotes an incoherence parameter with
$
\mu_2 \triangleq \max_{1\leq n \leq N} \left(  \sum_{k=1}^K |\phi_{nk}|^2 \right)\frac N K,
$
$\LL(M,r,f)$ is a function of $M$, $r$, $f$ and defined as
$
\LL(M,r,f) \triangleq \min\limits_{1\leq k \leq K} \frac{1}{r_k} \left[ \gamma_M(r_k)-\frac{c_2}{ \Delta_f} (1+r_k^{2M}) \right]
$
with $c_2$ being a constant and
\begin{align*}
\gamma_M(r_k) \triangleq \begin{cases}
\frac{r_k^{2M}-1}{2 \log(r_k)},~~&r_k < 1,\\
M, &r_k = 1,
\end{cases}
\end{align*}
and $\Delta_f \triangleq \min_{k\neq l} |f_k-f_l|$ denotes the minimum separation between true frequencies, where $|f_k-f_l|$ is the wrap-around distance on the unit circle. Then, $\Xs$ is the unique solution of~\eqref{NNMmiss} with probability at least $1-(MN)^{-2}$.
\textcolor{black}{Moreover, when $\textcolor{black}{\Xh =} \Xs$ is the unique solution to~\eqref{NNMmiss}, the dual solution $\Q$ has the form $\Q = \U_{\Xs}\V_{\Xs}^H + \W$, where $\U_{X^\star}^H \W = \zero,~\W\V_{X^\star} = \zero,~\text{and}~\|\W\| \leq 1$. In this case, if $\|\W\| < 1$, then the dual polynomial satisfies \begin{align*}
\|\QQ(r_k,f_k)\|_2 &= 1,~\forall~(r_k,f_k)\in\mathcal{RF}, \\
\|\QQ(r,f)\|_2 &< 1,~\forall~ (r,f)\notin\mathcal{RF}.
\end{align*}
}
\end{Corollary}

The proof for Corollary~\ref{THMdualploymiss} relies on some of the results in~\cite{chen2014robust}. However, those results do not extend directly to the damped exponential case. \textcolor{black}{Rather than use a certain incoherence property as in~\cite{chen2014robust}, we incorporate the damping ratios $\{r_k\}$ into the signal and develop theoretical guarantees that explicitly depend on the parameters, i.e., damping ratios and minimum frequency separation.}
In particular, we explicitly
bound the minimal singular value of $\At_{rf}^H \At_{rf}$ with the function $\LL(M,r,f)$ by exploiting the Vandermonde structure of $\At_{rf}$~\cite{aubel2017vandermonde}, instead of giving an incoherence property depending on just  the minimal singular value of $\At_{rf}^H \At_{rf}$ as in~\cite{chen2014robust}. Note that $\LL(M,r,f) = M-\frac{2c_2}{ \Delta_f} $ is on the order of $M$ when there is no damping (i.e., $r=1$) and the frequencies are well separated ($\Delta_f = \OO(\frac 1 M)$).
In the case where $r<1$, the sample complexity $|\Omega|$ scales inversely with $\LL(M,r,f) = \min\limits_{1\leq k \leq K} \frac{1}{r_k} \left[ \frac{r_k^{2M}-1}{2 \log(r_k)}-\frac{c_2}{ \Delta_f} (1+r_k^{2M}) \right]$, which increases monotonically with $r_k$ when the constant $c_2$ is sufficiently small. Therefore, it can be seen that the sample complexity $|\Omega|$ decreases if the minimum frequency separation $\Delta_f$ increases or the the damping ratio $r_k$ increases.
 Please see Section~\ref{proofTHMdualploymiss} for details  and Section~\ref{test_coherence} for supporting experiments. 

Note that the set $\Omega$ is chosen uniformly at random from all subsets of $\{1, \ldots, M\} \times \{1,\ldots, N \}$ with a given cardinality $|\Omega|$.
Since the columns of $\bPhi\in\CCC^{N\times K}$ are assumed to be normalized, we have $1\leq\mu_2\leq N$ according to the definition of $\mu_2$. In particular, $\mu_2=1$ when all the entries of $\bPhi$ have magnitude $\frac{1}{\sqrt{N}}$ and $\mu_2 = N$ when $\bPhi$ has a row containing all $1$'s with all other rows being $0$. Due to the normalized columns in $\bPhi$, we also have that $\sigma^2_{\min}(\bPhi) \leq 1$. Therefore, it can be seen from the above theorem that when there is no damping (or only light damping, i.e., $r_k$ is close to 1) and the frequencies are well separated, and $\mu_2$ is close to $1$, we can bound $\mu_1$ by a constant and thus the number of measurements needed for perfect recovery is comparable to best case bounds for rank-$K$ matrix completion. Specifically, state-of-the-art bounds~\cite{chen2015completing} for low-rank matrix completion from uniform random samples involve a dependence on a certain coherence parameter (equal to the maximum leverage score of the matrix); when this coherence parameter is small, the sample complexity is $|\Omega| = \OO(\max(M,N) K)$ (up to logarithmic factors). The significance to Corollary~\ref{THMdualploymiss} is that the sample complexity is not stated in terms of the matrix coherence; rather, the dependence on the damping ratios and minimum frequency separation is explicitly revealed.
\textcolor{black}{We also note that this is quite distinct from the work~\cite{chen2014robust}, in which the theoretical guarantees are built on a different incoherence property rather than the explicit parameters such as frequencies.}

Inspired by Algorithm~\ref{alg_NN-MUSIC} and the above analysis, we propose the following Missing Data MUSIC algorithm, named MD-MUSIC, to identify the damping factors $\{r_k\}$ and frequencies~$\{f_k\}$ from the partially observed data matrix $\Xs_{\Omega}$. Note that any off-the-shelf SDP solver could be used to solve the SDP in~(\ref{SDPmiss})\footnote{\textcolor{black}{
Primal-dual algorithms that are used in solvers such as CVX can return both the data matrix $\Xh$ and the dual solution $\Q$ of~\eqref{SDPmiss} and~\eqref{dual_prob}, respectively.}}.

\begin{algorithm}[H]
\caption{MD-MUSIC}\label{alg_MD-MUSIC}
\begin{algorithmic}[1]
\Procedure{Input}{$\Xs_{\Omega} \in \CCC^{M \times N}$}
\State compute $\Xh$ and $\Q$ by solving the SDP~(\ref{SDPmiss})
\State form the dual polynomial: $\QQ(r,f) = \Q^H \a(r,f)$
\State localize the places where $\|\!\QQ(r,f)\!\|_2 \!=\! 1$ to get $(\widehat{r}_k,\widehat{f}_k)$
\State \textbf{return} $\Xh$ and $(\widehat{r}_k,\widehat{f}_k)$
\EndProcedure
\end{algorithmic}
\end{algorithm}

Finally, we note that one could also consider an alternative approach wherein one first solves the NNM problem in~(\ref{SDPmiss}) and then uses Algorithm~\ref{alg_NN-MUSIC} to identify the $r_k$'s and $f_k$'s using $\Xh$. Interestingly, however, as we demonstrate in Theorem~\ref{THMdataxrfd}, it is sometimes possible with MD-MUSIC to perfectly recover the $r_k$'s and $f_k$'s even when exact recovery of $\Xs$ fails. This implies that MD-MUSIC is actually more powerful than the alternative approach mentioned above. We also conduct  simulations to further present this interesting phenomenon (parameter recovery without exact data matrix recovery) in Section~\ref{nume}.

\section{Proofs}
\label{proof}

\subsection{Proof for Theorem~\ref{THMdualploy}}
\label{proofTHMdualploy}
Denote a truncated SVD of $\A_{rf}$ as $\A_{rf} = \U_{A_{rf}} \S_{A_{rf}} \V_{A_{rf}}^H$. We first consider the case when $\W = \zero$, i.e., $\Q = \U_{X^\star} \V_{X^\star}^H$. Note that
\begin{align*}
&\|\Q^H\a(r,f)\|_2^2 \!=\! \|\U_{X^\star}^H\a(r,f)\|_2^2
 \!=\! \a(r,f)^H \U_{X^\star} \U_{X^\star}^H \a(r,f)
=\a(r,f)^H\U_{X^\star} \S_{X^\star} \V_{X^\star}^H \V_{X^\star} \S_{X^\star}^{-1}   \U_{X^\star}^H \a(r,f) \\
=& \a(r,f)^H\Xs (\Xs)^\dagger \a(r,f)
 = \a(r,f)^H \A_{rf} \D_{\ct} \bPhi^\top {(\bPhi^\top)}^\dagger \D_{\ct}^{-1} \A_{rf}^{\dagger} \a(r,f) 
 = \a(r,f)^H \A_{rf} \A_{rf}^{\dagger} \a(r,f) \\
=&  \a(r,f)^H \U_{A_{rf}} \U_{A_{rf}}^H \a(r,f)
 = \left\langle \PP_{\U_{A_{rf}}} (\a(r,f)), \a(r,f) \right\rangle,
\end{align*}
where $ \PP_{\U_{A_{rf}}} (\a(r,f)) \triangleq \U_{A_{rf}} \U_{A_{rf}}^H \a(r,f)$ is defined as the orthogonal projection of $\a(r,f)$ onto the range space of $\U_{A_{rf}}$, i.e., $\RR(\U_{A_{rf}})$.
Note that the third equality is obtained by plugging in $\I = \S_{X^\star} \V_{X^\star}^H \V_{X^\star} \S_{X^\star}^{-1}$ while the fifth equality is obtained by plugging in $\Xs= \A_{rf} \D_{\ct} \bPhi^\top$ and $(\Xs)^\dagger  =(\bPhi^\top)^\dagger \D_{\ct}^{-1} \A_{rf}^\dagger$. Also, note that $\bPhi^\top {(\bPhi^\top)}^\dagger = \I$
when $\bPhi\in\CCC^{N\times K}$ is of rank $K$,
which gives the sixth equality. The seventh equality holds due to
$
\A_{rf} \A_{rf}^{\dagger} = \U_{A_{rf}} \S_{A_{rf}} \V_{A_{rf}}^H \V_{A_{rf}} \S_{A_{rf}}^{-1} \U_{A_{rf}}^H =  \U_{A_{rf}}  \U_{A_{rf}}^H.
$

\begin{itemize}
\item For all $(r_k,f_k)\in\mathcal{RF}$, we have $\a(r_k,f_k) \in \RR(\U_{A_{rf}} )$, which implies
$
\PP_{\U_{A_{rf}}} (\a(r_k,f_k)) = \a(r_k,f_k).
$
Therefore, we have
$\|\Q^H\a(r_k,f_k)\|_2^2 = \left\langle \a(r_k,f_k), \a(r_k,f_k) \right\rangle = 1.$

\item For all $(r,f)\notin\mathcal{RF}$, if we have $\a(r,f) \notin \RR(\U_{A_{rf}} )$, which implies
$\PP_{\U_{A_{rf}}} (\a(r,f)) = \a(r,f) - \PP_{\U_{A_{rf}}^\perp } (\a(r,f)),$
we would then have
\begin{align*}
&\|\Q^H\a(r,f)\|_2^2 
= \left\langle \a(r,f), \a(r,f) \right\rangle -  \left\langle \PP_{\U_{A_{rf}}^\perp } (\a(r,f)), \a(r,f) \right\rangle 
<  \left\langle \a(r,f), \a(r,f) \right\rangle =1.
\end{align*}
\end{itemize}

Thus, we only need to show $\a(r,f) \notin \RR(\U_{A_{rf}} )$ for all $(r,f)\notin\mathcal{RF}$. Define a Vandermonde matrix $\A_{rf}^v \in \CCC^{M \times K}$ as
\begin{align}
\A_{rf}^v \triangleq [\a^v(r_1,f_1),\cdots,\a^v(r_K,f_K)] \label{Avrf}
\end{align}
with $\a^v(r,f) \triangleq [1~r e^{j2\pi f 1} ~\cdots~r^{M-1}e^{j 2\pi f (M-1)}]^\top$, which is the unnormalized version of $\a(r,f)$.
Then, $\A_{rf}$ is the column-normalized version of $\A_{rf}^v$. Assuming $M \geq K$, it follows that the first $K$ rows of $\A_{rf}^v$ form a square Vandermonde matrix, denoted as $\A_{rf}^K$, whose determinant is given by~\cite{klinger1967classroom, kalman1984generalized}
\begin{align*}
\change{det}(\A_{rf}^K) = \prod_{1\leq i<k\leq K} (r_k e^{j2\pi f_k} - r_i e^{j2\pi f_i} ).
\end{align*}
Then, $\change{rank}( \A_{rf}) =\change{rank}( \A_{rf}^v) = K$ if $M\geq K$ and $(r_i,f_i) \neq (r_k,f_k)$ for all $i \neq k$. Similarly, we have
\begin{align*}
\change{rank}( [\A_{rf}| \a(r,f)]) = K+1,
\end{align*}
i.e., $\a(r,f)\notin\RR(\A_{rf}) = \RR(\U_{\A_{rf}})$ if $(r,f)\notin\mathcal{RF}$, $M\geq K+1$ and all the $(r_k,f_k)$ pairs in $\mathcal{RF}$ are distinct. 

\textcolor{black}{It remains to show the case when $\W \neq \zero$. In particular, we have
$
\|\QQ(r,f)\|_2^2 = \|\U_{X^\star}^H\a(r,f)\|_2^2 + \|\W^H\a(r,f)\|_2^2.
$
\begin{itemize}
\item For all $(r_k,f_k)\in\mathcal{RF}$, we have $\a(r_k,f_k) \in \RR(\U_{A_{rf}}) = \RR(\U_{X^\star})$, and so together with $\U_{X^\star}^H \W = \zero$, we have
$
\|\QQ(r_k,f_k)\|_2^2  = \|\U_{X^\star}^H\a(r_k,f_k)\|_2^2 = 1.
$
\item For all $(r,f)\notin\mathcal{RF}$, we have  $\a(r,f) \notin \RR(\U_{A_{rf}})= \RR(\U_{X^\star})$ as shown above. Denote $\W = \U_{X^\star}^\perp \S_W {\V_{X^\star}^\perp}^H$ as a truncated SVD of $\W$ with $\U_{X^\star}^H \U_{X^\star}^\perp = \zero$, $\V_{X^\star}^H \V_{X^\star}^\perp = \zero$ and $\|\S_W\| = \|\W\| \leq 1$. In the case where $\|\S_W\| = \|\W\| < 1$, we have
\begin{align*}
\|\QQ(r,f)\|_2^2 &= \|\U_{X^\star}^H\a(r,f)\|_2^2 + \|\S_W {\U_{X^\star}^\perp}^H \a(r,f)\|_2^2	
\leq \|\U_{X^\star}^H\a(r,f)\|_2^2 + \|\S_W\|^2 \|{\U_{X^\star}^\perp}^H \a(r,f)\|_2^2\\
& < \|\U_{X^\star}^H\a(r,f)\|_2^2 + \|{\U_{X^\star}^\perp}^H \a(r,f)\|_2^2
 \leq 1.
\end{align*}
\end{itemize}
}
This completes the proof of Theorem~\ref{THMdualploy}.

\subsection{Proof for Theorem~\ref{THMdataxrfd}}
\label{proofTHMdataxrfd}

\textcolor{black}{Unlike in Corollary~\ref{THMdualploymiss}, we now focus on the case when the data matrix is not perfectly recovered, i.e., $\Xh \neq \Xs$. Again, with some fundamental Lagrange analysis, we can conclude that the dual solution $\Q = \Q_\Omega$ with minimal energy belongs to the subdifferential of $\|\Xh\|_*$. Then, we have $\Q = \Q_\Omega = \U_{\Xh}\V_{\Xh}^H + \W$ with $\U_{\Xh}^H \W = \zero,~\W\V_{\Xh} = \zero,~\text{and}~\|\W\| \leq 1$ for some matrix $\W$. Note that $\Q = \Q_\Omega$ is a sparse matrix with zero entries on the complement index of $\Omega$, so $\W$ should be a non-zero matrix. Recall that $\Xs = \A_{rf} \D_{\ct}\bPhi^\top$ with $\A_{rf} = [\a(r_1,f_1),\cdots,\a(r_K,f_K)]$. Then, we have 
\begin{align*}
\a(r_k,f_k) \in \RR(\A_{rf}) = \RR(\Xs) \subseteq \RR(\Xh) = \RR(\U_{\Xh}), \quad k = 1,\ldots, K.
\end{align*}   
\begin{itemize}
\item For all $(r_k,f_k)\in\mathcal{RF}$, we have $\a(r_k,f_k) \in \RR(\U_{\Xh})$, and so together with $\U_{\Xh}^H \W = \zero$, we have
$
\|\QQ(r_k,f_k)\|_2^2  = \|\U_{\Xh}^H\a(r_k,f_k)\|_2^2 = 1.
$
\item For all $(r,f)\notin\mathcal{RF}$, we have  $\a(r,f) \notin \RR(\A_{rf})= \RR(\Xs)$ as shown in Section~\ref{proofTHMdualploy}.
(1) If $\a(r,f) \in \RR(\U_{\Xh})$, we still have $
\|\QQ(r,f)\|_2^2  = 1
$ as above. This could result in false estimation of $(r,f)$ pairs when we localize the places where $\|\QQ(r,f)\|_2$ achieves 1.  
(2) If $\a(r,f) \notin \RR(\U_{\Xh})$, denote $\W = \U_{\Xh}^\perp \S_W {\V_{\Xh}^\perp}^H$ as a truncated SVD of $\W$ with $\U_{\Xh}^H \U_{\Xh}^\perp = \zero$, $\V_{\Xh}^H \V_{\Xh}^\perp = \zero$ and $\|\S_W\| = \|\W\| \leq 1$. Recall that $\W \neq \zero$. In the case where $\|\S_W\| = \|\W\| < 1$, we have
\begin{align*}
\|\QQ(r,f)\|_2^2 &= \|\U_{\Xh}^H\a(r,f)\|_2^2 + \|\S_W {\U_{\Xh}^\perp}^H \a(r,f)\|_2^2	
\leq \|\U_{\Xh}^H\a(r,f)\|_2^2 + \|\S_W\|^2 \|{\U_{\Xh}^\perp}^H \a(r,f)\|_2^2\\
& < \|\U_{\Xh}^H\a(r,f)\|_2^2 + \|{\U_{\Xh}^\perp}^H \a(r,f)\|_2^2
 \leq 1.
\end{align*}
\end{itemize}
Thus, we finish the proof of Theorem~\ref{THMdataxrfd}.
 }

\subsection{Proof for Corollary~\ref{THMdualploymiss}}
\label{proofTHMdualploymiss}

Define
$
\D_{cM} \triangleq \sqrt{M} \change{diag}([c_1,~c_2,~\cdots,~c_K]),~
\At_{rf} \triangleq \frac{1}{\sqrt{M}} \A_{rf}^v,
$
where $\A_{rf}^v$ is the unnormalized Vandermonde matrix and defined in~\eqref{Avrf}.
Observe that the transpose of the noiseless data matrix $\Xs_\top \triangleq  {\Xs}^\top = \bPhi \D_{\ct} \A_{rf}^\top = \bPhi \D_{cM} \At_{rf}^\top$ can be viewed as the block Hankel matrix $\X_e$ introduced in~\cite{chen2014robust}, but with $k_1 = N$ and $k_2 = 1$.
Define $\PP_T$ as the projection operator that acts on the tangent space of $\Xs_\top$. Denote a truncated SVD of $\Xs_\top$ as $\Xs_\top = \U_{\Xs_\top}  \S_{\Xs_\top} \V_{\Xs_\top}^H$.
Then, \cite[Lemma 1]{chen2014robust} can be adapted to provide us sufficient conditions that are used to guarantee the unique optimality of $\Xs$. In particular, we can set $\AA$ and $\AA_\Omega$ as in \cite[Lemma 1]{chen2014robust} as the identity operator and the random sampling operator $\PP_\Omega$, respectively. Therefore, we need the following condition
\begin{align}
\left\|\PP_T - \frac{NM}{|\Omega|} \PP_T \PP_\Omega \PP_T\right\| \leq \frac 1 2.
\label{cond_PT}
\end{align}

%
%

Next, we verify that the above condition~\eqref{cond_PT} holds with high probability under certain conditions. Define $\A_{(n,m)} \in \RRR^{N\times M}$ as a matrix with the $(n,m)$-th entry being 1 and others being 0. We first quantify the projection of $\A_{(n,m)}$ onto the subspace $T$, the tangent space of $\Xs_\top$. In particular, we have the following lemma which utilizes a quite different incorence property than the one used in~\cite[Lemma~2]{chen2014robust}.
\begin{Lemma}\label{lem_incoherence}
For some constant $\mu_1$, if
\begin{align}
\sigma_{\min}(\bPhi^H \bPhi) \geq \frac{\mu_2}{\mu_1} \quad\text{and}\quad \sigma_{\min}(\At_{rf}^H \At_{rf}) \geq \frac{1}{\mu_1}, \label{incoherence}
\end{align}
then
\begin{align*}
\|\U_{\Xs_\top} \U_{\Xs_\top}^H \A_{(n,m)}\|_F^2 \leq \frac{\mu_1 c_s K}{NM}, \quad
\| \A_{(n,m)} \V_{\Xs_\top} \V_{\Xs_\top}^H \|_F^2 \leq \frac{\mu_1 c_s K}{NM}
\end{align*}
hold for any $(n,m)\in [N]\times [M]$ with $[N]\triangleq \{1,2,\ldots,N\}$ and $[M]\triangleq \{1,2,\ldots,M\}$.
We have defined $c_s \triangleq \max\{N,M\}$ and $\mu_2 \triangleq \max_{1\leq n \leq N} \left(  \sum_{k=1}^K |\phi_{nk}|^2 \right)\frac N K$. It follows that
\begin{equation}
\begin{aligned}
\|\PP_T(\A_{(n,m)})\|_F^2 
\leq  \|\U_{\Xs_\top} \U_{\Xs_\top}^H \A_{(n,m)}\|_F^2  +  \| \A_{(n,m)} \V_{\Xs_\top} \V_{\Xs_\top}^H \|_F^2  
\leq  \frac{2\mu_1 c_s K}{NM}. \label{eq_lemma1}
\end{aligned}
\end{equation}
\end{Lemma}

\begin{proof}
Note that $\U_{\Xs_\top}~(\V_{\Xs_\top})$ and $\bPhi~(\At_{rf})$ determine the same column (row) space of $\Xs_\top$. In particular, we have
\begin{align*}
\U_{\Xs_\top} \U_{\Xs_\top}^H = \bPhi (\bPhi^H \bPhi)^{-1} \bPhi^H,  \quad
\V_{\Xs_\top} \V_{\Xs_\top}^H = \At_{rf} (\At_{rf}^H \At_{rf})^{-1} \At_{rf}^H,
\end{align*}
which implies
\begin{align*}
\|\U_{\Xs_\top} \U_{\Xs_\top}^H \A_{(n,m)}\|_F^2 
= &\| \bPhi (\bPhi^H \bPhi)^{-1} \bPhi^H  \A_{(n,m)}\|_F^2
=  \left\langle  (\bPhi^H \bPhi)^{-1} \bPhi^H \A_{(n,m)},  \bPhi^H \A_{(n,m)} \right\rangle \\
\leq &\| (\bPhi^H \bPhi)^{-1} \| \| \bPhi^H \A_{(n,m)} \|_F^2
= \frac{1}{\sigma_{\min}( \bPhi^H \bPhi )} \| \bPhi^H \A_{(n,m)} \|_F^2
\end{align*}
and
\begin{align*}
\| \A_{(n,m)} \V_{\Xs_\top} \V_{\Xs_\top}^H \|_F^2 
= &\| \A_{(n,m)} \At_{rf} (\At_{rf}^H \At_{rf})^{-1} \At_{rf}^H \|_F^2
=  \left\langle \A_{(n,m)} \At_{rf} (\At_{rf}^H \At_{rf})^{-1}, \A_{(n,m)} \At_{rf} \right\rangle \\
 \leq & \| (\At_{rf}^H \At_{rf})^{-1} \| \| \A_{(n,m)} \At_{rf} \|_F^2
= \frac{1}{\sigma_{\min}(\At_{rf}^H \At_{rf})} \| \A_{(n,m)} \At_{rf}  \|_F^2.
\end{align*}

Define
$
\mu_2 \triangleq \max_{1\leq n \leq N} \left(  \sum_{k=1}^K |\phi_{nk}|^2 \right)\frac N K.
$
Note that $1\leq \mu_2 \leq N$.
Recall that $\A_{(n,m)} \in \RRR^{N\times M}$ is a matrix with the $(n,m)$-th entry being $1$ and all others being $0$. Therefore, we can bound $\| \bPhi^H \A_{(n,m)} \|_F^2$ and $\| \A_{(n,m)} \At_{rf}  \|_F^2$ with
\begin{align*}
\| \bPhi^H \A_{(n,m)} \|_F^2 = \sum_{k=1}^K |\phi_{nk}|^2 \leq \mu_2 \frac K N,\quad
\| \A_{(n,m)} \At_{rf}  \|_F^2 = \sum_{k=1}^K \frac{1}{M} r_k^{2m} \leq \frac K M.
\end{align*}
Define $c_s \triangleq \max(N,M)$. Then, if
\begin{align*}
\sigma_{\min}( \bPhi^H \bPhi ) \geq \frac{\mu_2}{\mu_1}, \quad\sigma_{\min}(\At_{rf}^H \At_{rf}) \geq \frac{1}{\mu_1},
\end{align*}
we can get
\begin{align*}
\|\U_{\Xs_\top} \U_{\Xs_\top}^H \A_{(n,m)}\|_F^2 & \leq  \frac{\mu_1 K}{N} =  \frac{\mu_1 K M}{NM} \leq  \frac{\mu_1 K c_s}{NM},  \\
\| \A_{(n,m)} \V_{\Xs_\top} \V_{\Xs_\top}^H \|_F^2 & \leq \frac{\mu_1 K}{M} =  \frac{\mu_1 K N}{NM} \leq  \frac{\mu_1 K c_s}{NM}.
\end{align*}
Then, we obtain~\eqref{eq_lemma1}.
\end{proof}

Similar to Lemma 3 in~\cite{chen2014robust}, we would then have that condition~\eqref{cond_PT} holds with probability at least $1-(NM)^{-4}$ if
$
|\Omega| \geq c_1 \mu_1 c_s K \log(NM),
$
where $c_1\geq 0$ is a constant.

The remaining proof for Corollary~\ref{THMdualploymiss} follows the corresponding proof steps for Theorem 1 in~\cite{chen2014robust}. This yields Corollary~\ref{THMdualploymiss}, which is similar to Theorem 1 in~\cite{chen2014robust} but with different incoherence properties~\eqref{incoherence}.

To obtain these different incoherence properties, we bound the minimum nonzero singular value of $\bPhi$ and $\At_{rf}$.
It follows from Theorem 5 of~\cite{aubel2017vandermonde} that
\begin{align*}
\sigma_{\min}(\At_{rf}^H \At_{rf})& = \sigma_{\min}^2(\At_{rf})
= \frac 1 M  \sigma_{\min}^2(\A_{rf}^v) 
\geq \frac 1 M  \LL(M,r,f),
\end{align*}
where  $\LL(M,r,f)$ is defined as
\begin{align*}
\LL(M,r,f) \triangleq \min\limits_{1\leq k \leq K} \frac{1}{r_k} \left[ \gamma_M(r_k)-\frac{c_2}{ \Delta_f} (1+r_k^{2M}) \right]
\end{align*}
with $c_2$ being a constant and
\begin{align*}
\gamma_M(r_k) \triangleq \begin{cases}
\frac{r_k^{2M}-1}{2 \log(r_k)},~~&r_k < 1,\\
M, &r_k = 1,
\end{cases}
~~~~~k=1,\ldots,K.
\end{align*}
Note that
$
\LL(M,r,f) = M-\frac{2c_2}{ \Delta_f}
$
is on the order of $M$ when there is no damping (i.e., $r=1$) and the frequencies are well separated ($\Delta_f = O(\frac 1 M)$).
To satisfy the two assumptions in~\eqref{incoherence}, we can let
$
\frac 1 M \LL(M,r,f)\geq
\frac{1}{\mu_1},~
\text{and}~\mu_1\geq \frac{\mu_2}{\sigma_{\min}^2(\bPhi)},
$
that is,
$
\mu_1\geq
\max \left\{\frac{M} {\LL(M,r,f)}, \frac{\mu_2}{\sigma_{\min}^2(\bPhi)} \right\}.
$

With some fundamental Lagrange analysis as in Section~\ref{full}, we can also conclude that the optimal dual solution $\Q$ belongs to the subdifferential of $\|\Xh\|_*$, where $\Xh$ is the optimal solution of~\eqref{NNMmiss}. In the case of exact recovery, i.e., $\Xh = \Xs$, the dual solution $\Q$ has the form $\Q = \U_{\Xs}\V_{\Xs}^H + \W$, where $\U_{X^\star}^H \W = \zero,~\W\V_{X^\star} = \zero,~\text{and}~\|\W\| \leq 1$. Then, we are still able to identify the true $r_k$'s and $f_k$'s by localizing the places where  $\|\QQ(r,f)\|_2 = \|\Q^H \a(r,f)\|_2$ achieves 1 with $\Q$ being the optimal dual solution of~\eqref{dual_prob} \textcolor{black}{and $\|\W\|<1$} in the case when $\Xh = \Xs$. \textcolor{black}{Thus, we finish the proof of Corollary~\ref{THMdualploymiss}.}

\section{Numerical Simulations}
\label{nume}

\subsection{Full data case}
\label{fdc}

In this experiment, we use synthetic data to test the proposed Algorithm~\ref{alg_NN-MUSIC} with $K = 3$. The true $r_k$'s and $f_k$'s are set as $r_1 = 0.92,~r_2 = 0.98,~r_3 = 0.85$ and $f_1 = 0.1,~f_2 = 0.4,~f_3 = 0.8$. We set $M=N =50$. The data matrix $\Xs$ is then generated as
$
\Xs = \A_{rf} \D_{\ct} \bPhi^\top,
$
where $\A_{rf},~\text{and}~~ \D_{\ct}$ are generated according to their definition in Section~\ref{prob}, $\bPhi$ is generated as a Gaussian random matrix with normalized columns, and the $c_k$'s are set as $K$ Gaussian random numbers with zero mean and unit variance. The first three columns of $\Xs$ are shown in Figure~\ref{test_NA_data} (a).
Given the above data matrix $\Xs$, we then use Algorithm~\ref{alg_NN-MUSIC} to identify all the $r_k$'s and $f_k$'s. Figure~\ref{test_NA_data} (b) displays a surface plot of $\|\QQ(r,f)\|_2$ and indicates that Algorithm~\ref{alg_NN-MUSIC}  identifies all the $r_k$'s and $f_k$'s perfectly.

Next, as a demonstration, we repeat the above experiment but with additive white Gaussian noise with variance $\sigma = 0.1$ (SNR = $9.8806$dB). The noiseless data and noisy data are shown in Figure~\ref{test_NA_data_noise}~(a). We set the regularization parameter\footnote{Here, the regularization parameter is set according to~\cite{li2017atomicIEEE}.} as $\lambda = 0.2\sigma \sqrt{4MN\log(M)}  = 3.9558$ and then solve the nuclear norm denoising program~\eqref{NNMnoise}. As is shown in Figure~\ref{test_NA_data_noise}~(b), we observe that the $(r,f)$ pairs can still be estimated by localizing the peaks of $\|\QQ(r,f)\|_2$.
In particular, the estimated damping ratios and frequencies are given as $\rh_1 = 0.92,~\rh_2 = 0.98,~\rh_3 = 0.79$ and $\fh_1 = 0.1,~\fh_2 = 0.4,~\fh_3 = 0.8$.
Note that we leave the corresponding theoretical guarantees for future work.

\begin{figure}[t]

\begin{minipage}{0.48\linewidth}
\centering
\includegraphics[width=2.3in]{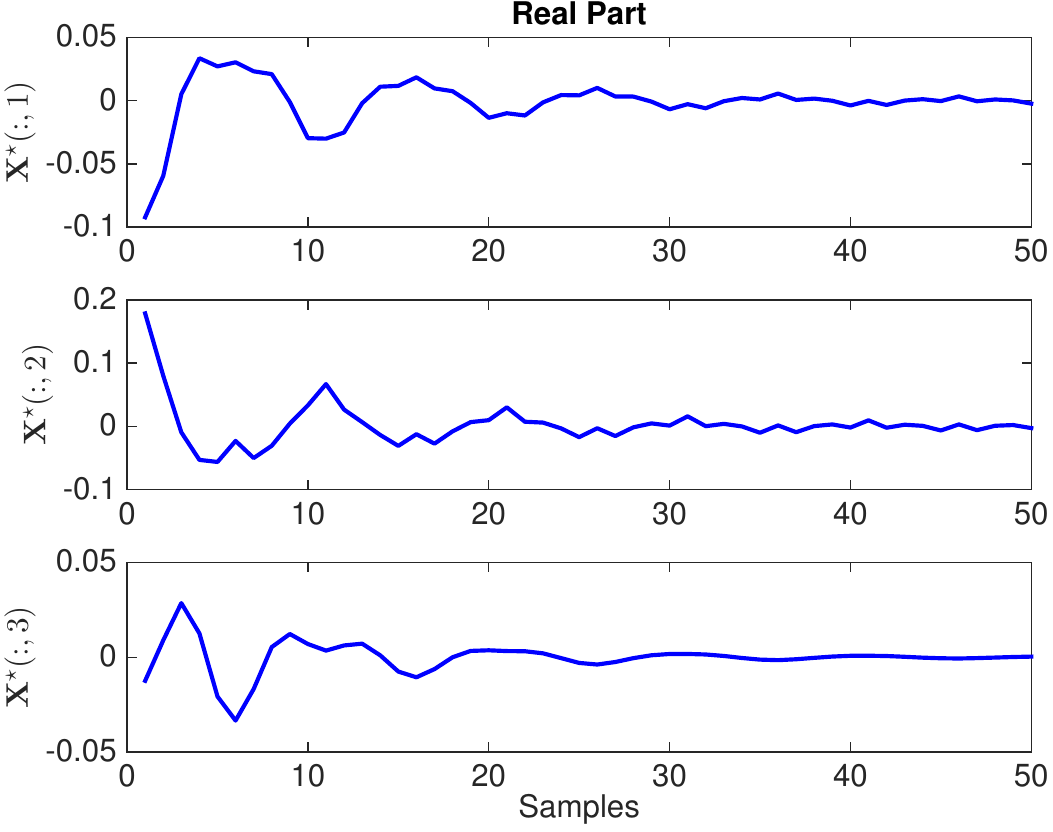}
\centerline{\small{(a)}}
\end{minipage}
\hfill
\begin{minipage}{0.48\linewidth}
\centering
\includegraphics[width=2.5in]{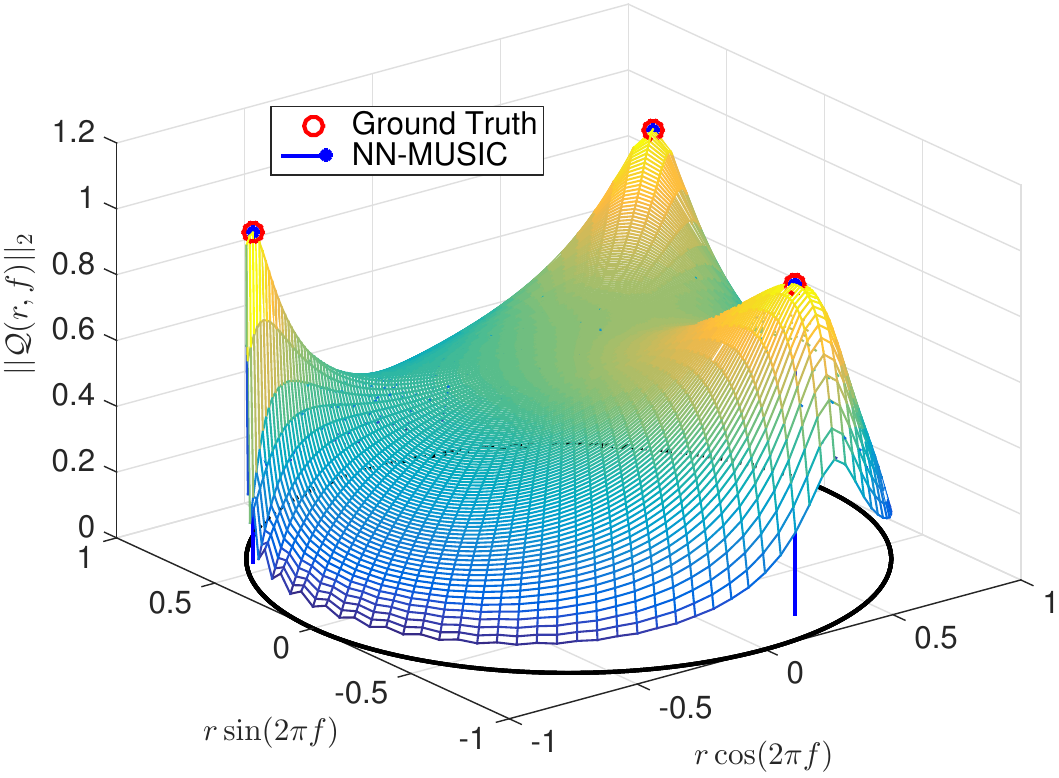}
\centerline{\small{(b) $40\%$ missing}}
\end{minipage}
\caption{Noiseless full data: (a) The first three columns in data matrix $\Xs$. (b) The blue lines correspond to the locations where $\|\QQ(r,f)\|_2$ achieves 1 while the red circles correspond to the true $r_k$'s and $f_k$'s. They coincide because the recovery is perfect.}
\label{test_NA_data}
\end{figure}


\begin{figure}[t]

\begin{minipage}{0.48\linewidth}
\centering
\includegraphics[width=2.3in]{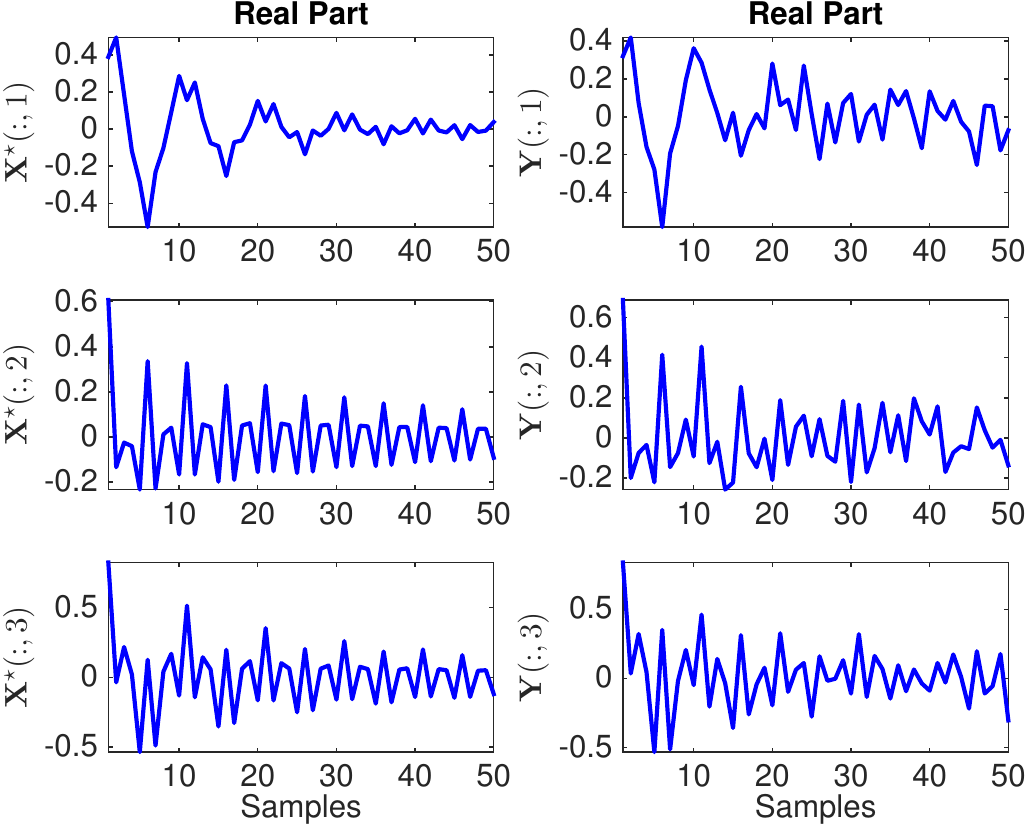}
\centerline{\small{(a)}}
\end{minipage}
\hfill
\begin{minipage}{0.48\linewidth}
\centering
\includegraphics[width=2.5in]{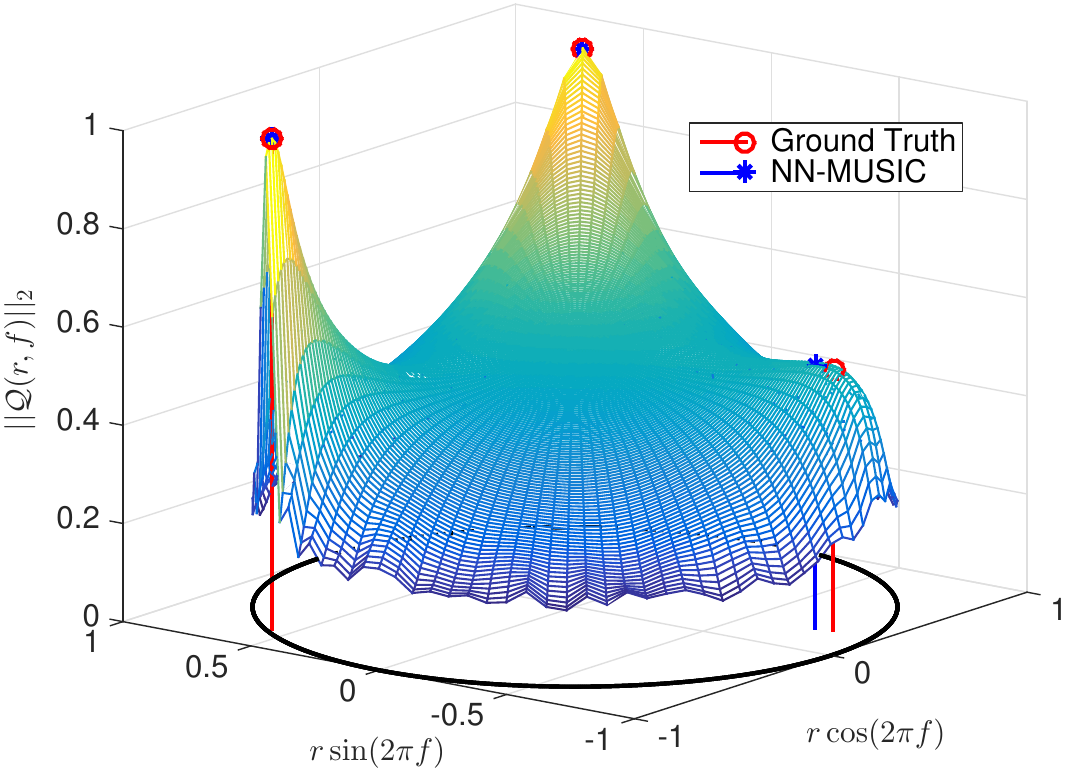}
\centerline{\small{(b) $40\%$ missing}}
\end{minipage}
\caption{Noisy full data: (a) Left: noiseless data $\Xs$. Right: noisy data $\Y$. (b) The blue lines correspond to the peaks of $\|\QQ(r,f)\|_2$ while the red lines mark the position of the true $r_k$'s and $f_k$'s.}
\label{test_NA_data_noise}
\end{figure}

%
%
%

\subsection{Missing data case}
\label{mdc}

We repeat the above experiments with missing data, namely we identify the damping factors and frequencies from the given partially observed data matrix $\Xs_{\Omega}$ by solving the NNM problem in~(\ref{NNMmiss}).\footnote{Note that CVX~\cite{grant2008cvx} can return the estimated data matrix $\Xh$ as well as the dual solution $\Q$ by solving the SDP form in~(\ref{SDPmiss}).} All parameters are set same as in the noise-free setting of Section~\ref{fdc}. After generating the full data matrix $\Xs$, we randomly remove $20\%$ and $40\%$ of its entries. We then use Algorithm~\ref{alg_MD-MUSIC} to identify all the $r_k$'s and $f_k$'s from the partial data.
Figure~\ref{test_MUSIC_Damping_Missing} indicates that Algorithm~\ref{alg_MD-MUSIC}  identifies all the $r_k$'s and $f_k$'s perfectly.

We notice that the data matrix $\Xs$ is also well recovered in this case. In particular, we define the relative recovery error of data matrix as $\change{RelErr}\triangleq\frac{\|\Xh-\Xs\|_F}{\|\Xs\|_F}$, where $\Xs$ and $\Xh$ denote the true full data matrix and the recovered data matrix via NNM. In particular, we have $\change{RelErr} = 4.7487 \times 10^{-10} $ when $20\%$ of the data is missing and $\change{RelErr} =3.5190\times 10^{-8} $ when $40\%$ of the data is missing.
Moreover, as is shown in Figure~\ref{test_MUSIC_Damping_Missing_dataerr}, we also observe that in some cases, the $r_k$'s and $f_k$'s can be perfectly recovered even if we do not perfectly recover $\Xs$, which further supports our Theorem~\ref{THMdataxrfd}.

Finally, we investigate the minimal number of measurements needed for perfect recovery with various numbers $K$ of spectral components. We set $M=70$ and $N=50$. For each value of $K$, we randomly pick $K$ frequencies and damping ratios from a frequency set $\FF = 0.05:0.05:0.95$ and a damping ratio set $\RR = 0.94:0.0025:1$.\footnote{We choose 0.94 as the lowest damping ratio since we want to keep at least $1\%$ energy at the end of uniform sampling. Therefore, we have $r=0.01^{\frac 1 M} \leq 0.94$.} Denote $(\rh,\fh)$ and $(\rs,\fs)$ as the recovered parameters and true parameters, respectively. We consider the parameter recovery to be a success if
\begin{align}
\max_{1\leq k \leq K}( |\rh_k-\rs_k|) \!\leq \!10^{-5},~\max_{1\leq k \leq K}( |\fh_k-\fs_k|) \! \leq \!10^{-5}.
\label{succ_para}
\end{align}
Similarly, we consider the data matrix recovery to be a success if the relative recovery error
$
\frac{\|\Xh-\Xs\|_F}{\|\Xs\|_F}  \leq 10^{-5}.
$
We perform 20 trials in this part of simulation. It can be seen in Figure~\ref{test_MUSIC_Damping_Missing_phasetransition} that the minimal number of measurements needed for perfect data matrix recovery does scale roughly linearly with $K$, as indicated in Corollary~\ref{THMdualploymiss}. We also notice a similar behavior appearing in parameters recovery. Figure~\ref{test_MUSIC_Damping_Missing_phasetransition} (c) again indicates that we can still successfully recover the parameters in some cases where the data matrix  is not perfectly recovered.

\begin{figure}[t]

\begin{minipage}{0.48\linewidth}
\centering
\includegraphics[width=2.7in]{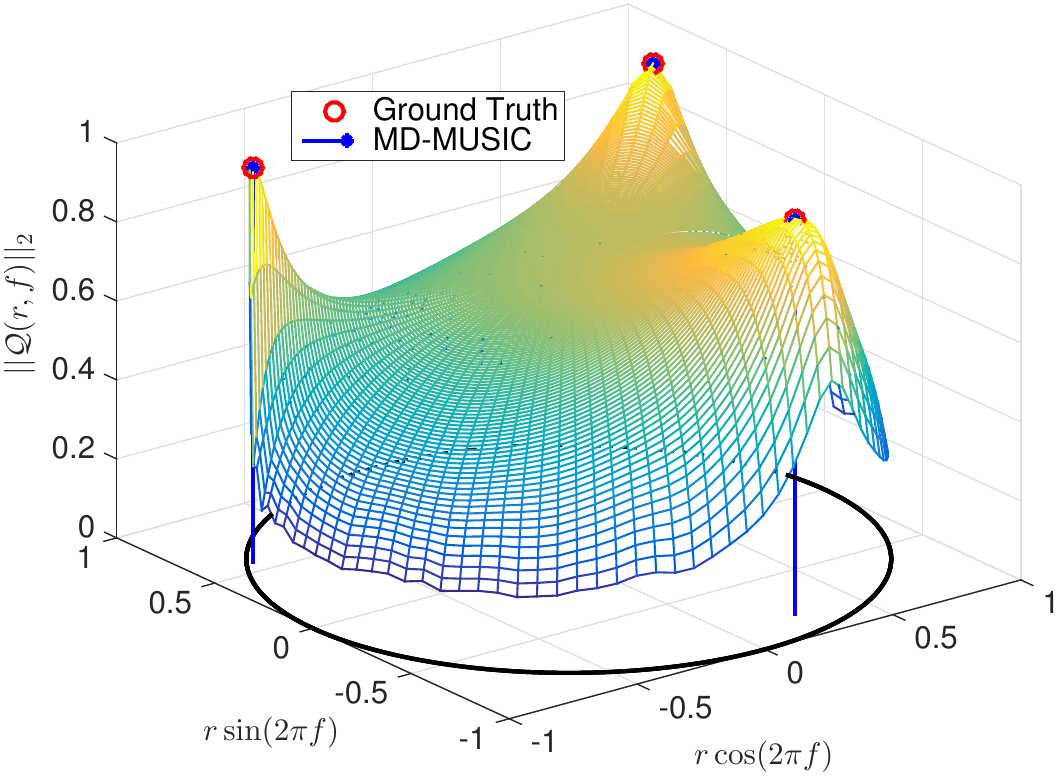}
\centerline{\small{(a) $20\%$ missing}}
\end{minipage}
\hfill
\begin{minipage}{0.48\linewidth}
\centering
\includegraphics[width=2.7in]{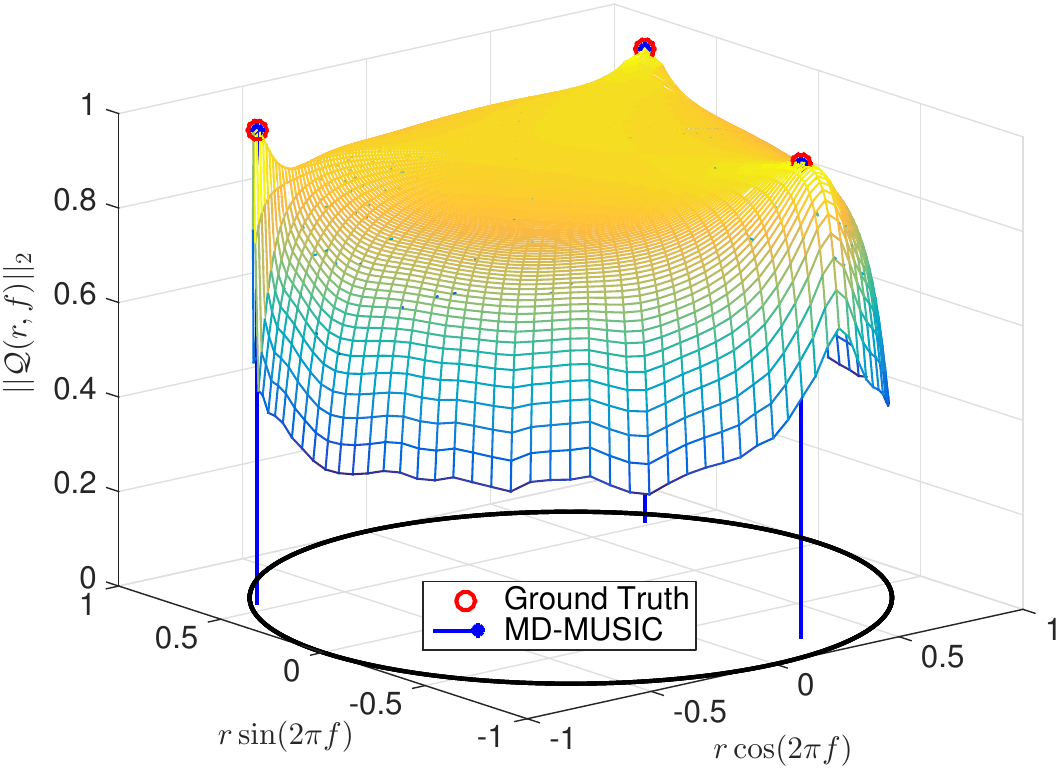}
\centerline{\small{(b) $40\%$ missing}}
\end{minipage}
\caption{Noiseless missing data: the blue lines correspond to the locations where $\|\QQ(r,f)\|_2$ achieves 1 while the red circles correspond to the true $r_k$'s and $f_k$'s. They coincide because the recovery is perfect.}
\label{test_MUSIC_Damping_Missing}
\end{figure}

\begin{figure}[t]

\begin{minipage}{0.48\linewidth}
\centering
\includegraphics[width=2.5in]{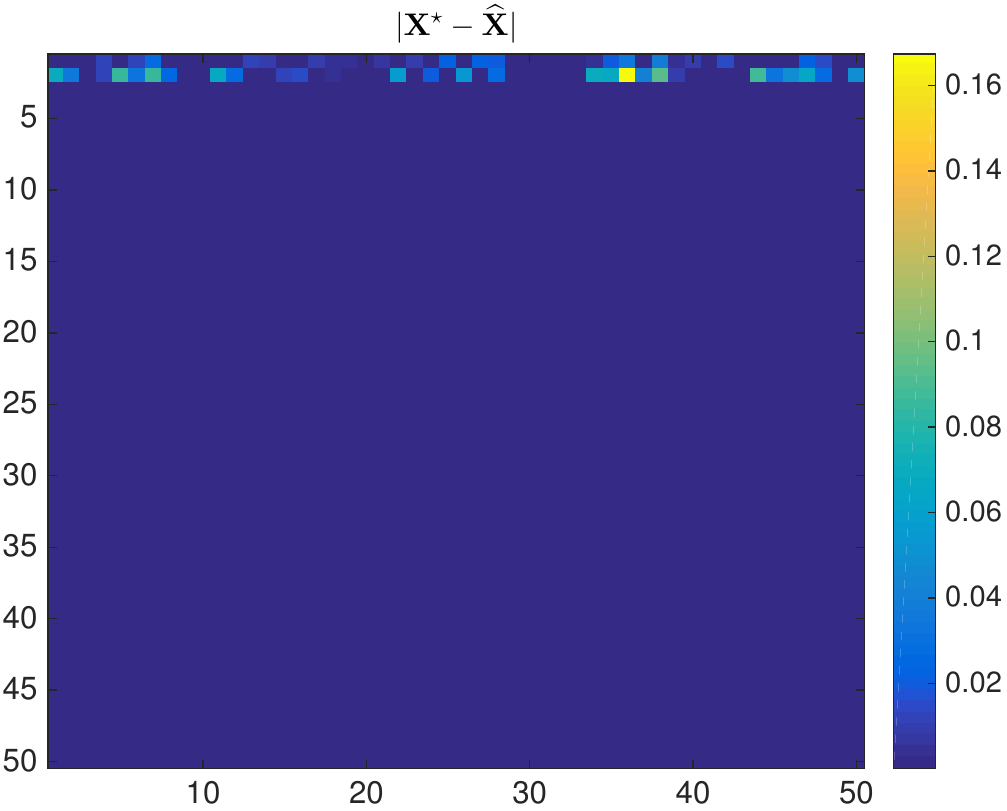}
\centerline{\small{(a) data reconstruction error}}
\end{minipage}
\hfill
\begin{minipage}{0.48\linewidth}
\centering
\includegraphics[width=2.7in]{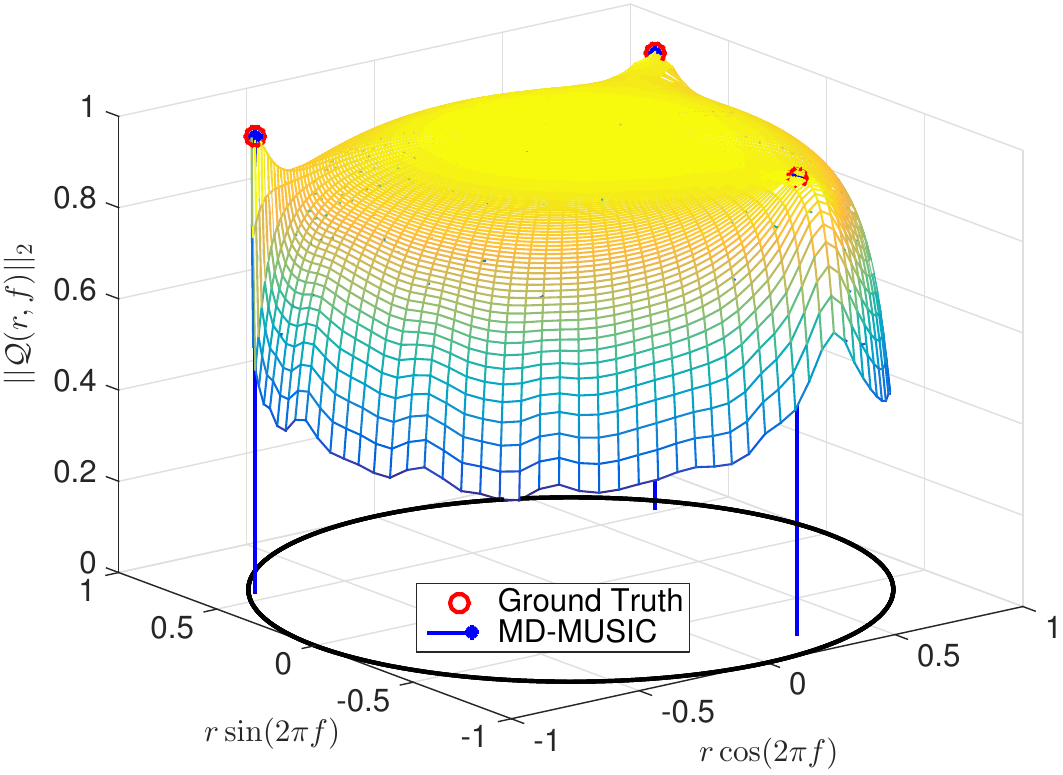}
\centerline{\small{(b) $40\%$ missing}}
\end{minipage}
\caption{Noiseless missing data: in some cases, the $r_k$'s and $f_k$'s can be perfectly recovered even when $\Xs$ is not. In this example, we have RelErr = 0.0412.}
\label{test_MUSIC_Damping_Missing_dataerr}
\end{figure}

\begin{figure}[t]
\begin{minipage}{0.31\linewidth}
\centering
\includegraphics[width=1.8in]{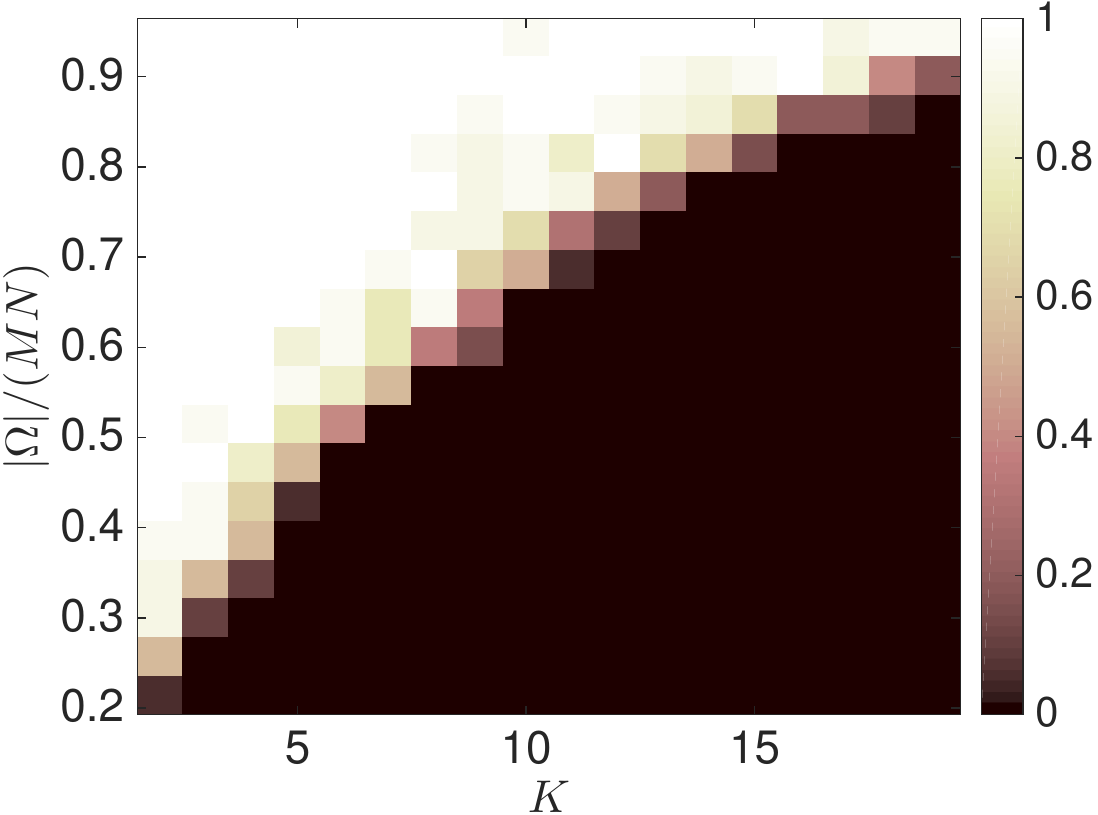}
\centerline{\footnotesize{(a) $(r,f)$}}
\end{minipage}
\hfill
\begin{minipage}{0.31\linewidth}
\centering
\includegraphics[width=1.8in]{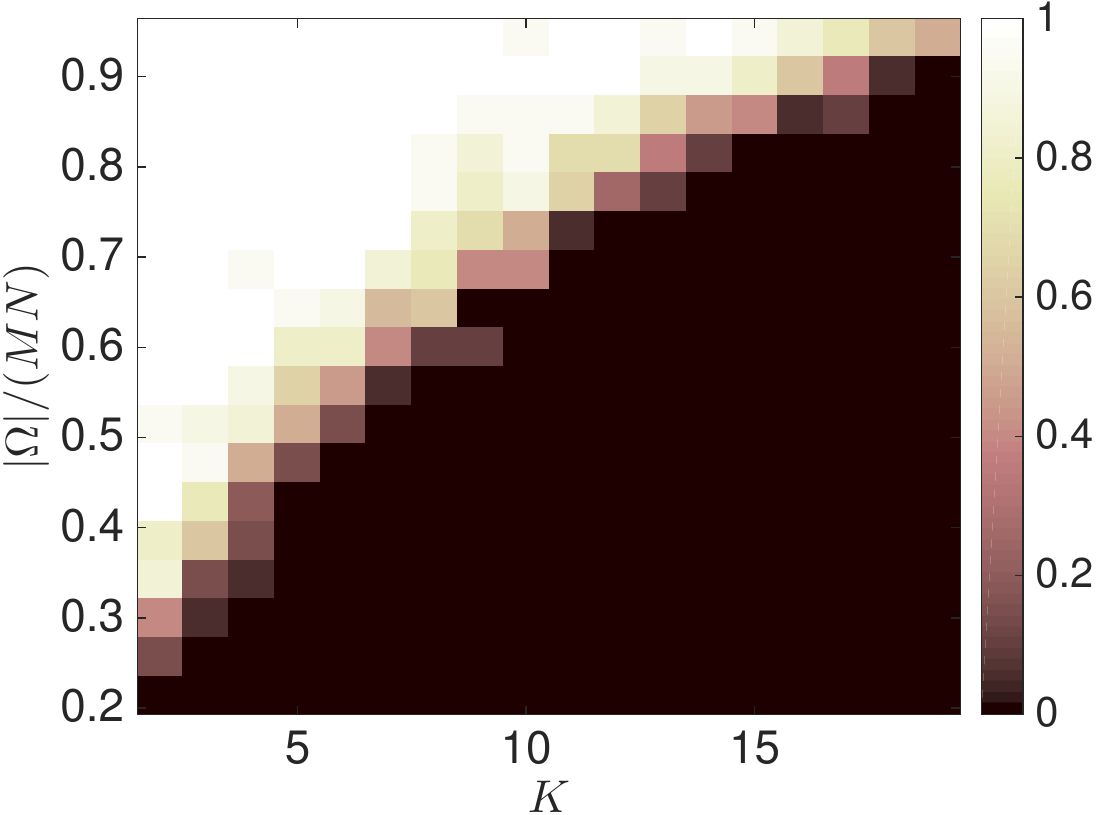}
\centerline{\footnotesize{(b) $\Xs$}}
\end{minipage}
\hfill
\begin{minipage}{0.31\linewidth}
\centering
\includegraphics[width=1.8in]{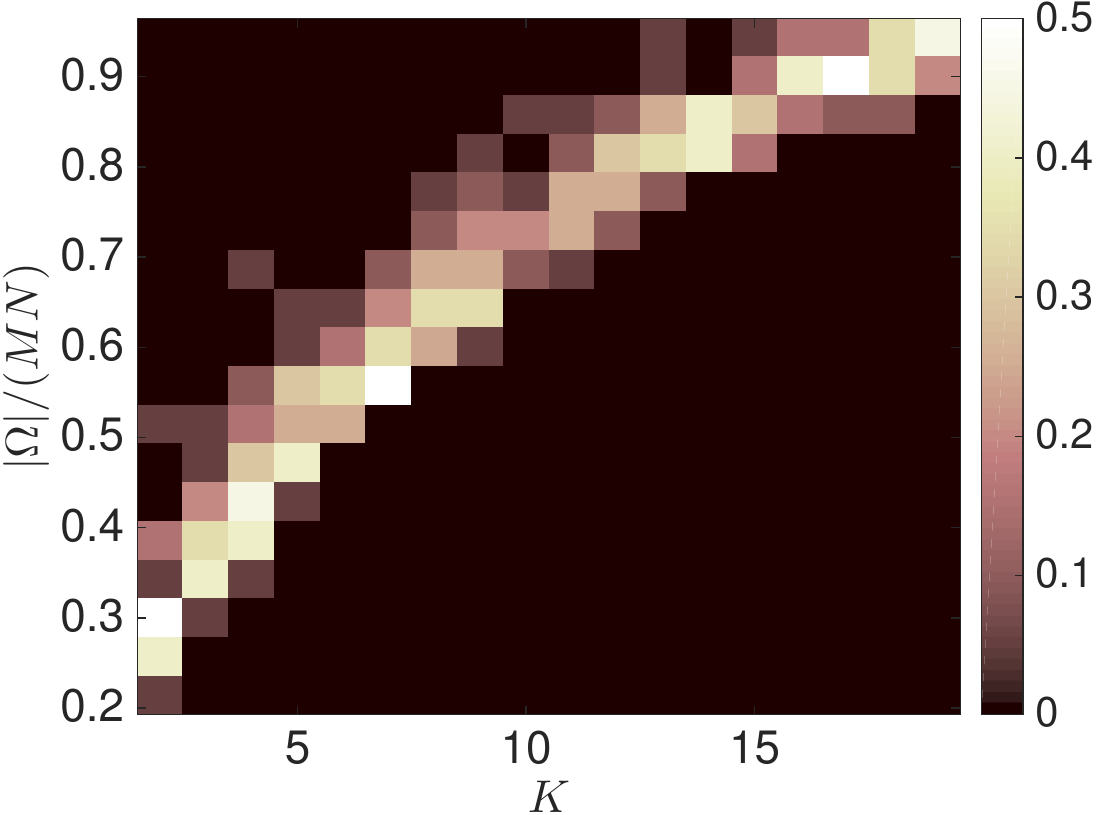}
\centerline{\footnotesize{(c) probability difference}}
\end{minipage}
\caption{Probability of successful recovery for (a) parameters $(r,f)$ and (b) data matrix $\Xs$. (c) presents the probability of successful recovery for parameters minus the probability of successful recovery for the data matrix.}
\label{test_MUSIC_Damping_Missing_phasetransition}
\end{figure}

\subsection{Data coherence}
\label{test_coherence}

In this section, we conduct three numerical experiments to examine the influence of the minimum frequency separation~$\Delta_f$, the matrix~$\bPhi$, and the damping ratio $r_k$ on the performance of missing data matrix recovery. The standard literature on matrix completion~\cite{chen2015completing} relates the recoverability of a matrix $\Xs$ to its {\em coherence}, defined as
$
\mu_0^\star \triangleq \max \left\{ \mu_1^\star(\Xs), \mu_2^\star(\Xs)  \right\}
$
with
\begin{align*}
\mu_1^\star(\Xs) \triangleq \frac{M}{K} \max_{1\leq m\leq M} \left\| \U_{\Xs}^\top \e_m^c \right\|_2^2,\quad
\mu_2^\star(\Xs) \triangleq \frac{N}{K} \max_{1\leq n\leq N} \left\| \V_{\Xs}^\top \e_n^c \right\|_2^2,
\end{align*}
where $\Xs = \U_{X^\star} \S_{X^\star} \V_{X^\star}^H$ is a truncated SVD of $\Xs$, and $\e_m^c \in \RRR^M$ and $\e_n^c\in\RRR^N$ denote canonical basis vectors.

In the first experiment, we examine the influence of minimum frequency separation on the performance of missing data recovery with $M=50$, $N=30$, $K=2$ and $|\Omega |= 450$, i.e., $70\%$ of the data are missing. To simplify the experiment, we set $r_1 = r_2 =1$ and $c_1 = c_2 = 1$.  We fix $f_1 = 0.1$ and let $f_2 = f_1+\Delta_f$ with various values of the minimum frequency separation $\Delta_f$. We generate $\bPhi\in\CCC^{N\times K}$ using normalized columns from a discrete Fourier matrix, which implies $\mu_2^\star = 1$ and ensures that $\mu_0^\star = \mu_1^\star$. $10^4$ trials are performed in this experiment.
Other settings are the same as in Section~\ref{mdc}. It is shown in Figure~\ref{test_MUSIC_Missing_co_pa_deltf} (a, b) that the coherence parameter $\mu_0^\star$ decreases as the minimum frequency separation $\Delta_f$ increases, which also explains why the probability of successful data matrix recovery increases as  $\Delta_f$ increases.

In the second experiment, we examine the influence of the matrix $\bPhi$ on the performance of missing data recovery. We change $N=10$ to make sure that the coherence parameter $\mu_0^\star$ is not too large. By fixing $\Delta_f = 1/M$, we have $\mu_1^\star = 1$ and thus $\mu_0^\star = \mu_2^\star$. We again generate $\bPhi\in\CCC^{N\times K}$ using columns from the discrete Fourier matrix, but we then replace its first entry $\phi_{1,1}$ with scalars in the range of $[1,10]$ and then normalize its columns. Other settings are same as the first experiment. We conduct 500 trials in this experiment. Figure~\ref{test_MUSIC_Missing_co_pa_deltf} (c, d) shows that the coherence parameter $\mu_0^\star$ increases as $\phi_{1,1}$ increases, which also explains why the probability of successful data matrix recovery decreases as $\phi_{1,1}$ increases.

In the third experiment, we examine the influence of the damping ratio on the performance of missing data recovery. We fix $\Delta_f = 1/M$ and repeat the first experiment with various values of $r_2$. We conduct 100 trials in this experiment. As shown in Figure~\ref{test_MUSIC_Missing_co_pa_deltf} (e, f), the coherence parameter $\mu_0^\star$ decreases as $r_2$ increases, which also explains why the probability of successful data matrix recovery increases as $r_2$ increases. This is to be expected since exponentials with smaller damping ratio are transient and their contribution tends to fade quickly from the measured data. Moreover, Figure~\ref{test_MUSIC_Missing_co_pa_deltf} (f) again indicates that we can successfully recover the parameters in some cases where the data matrix is not perfectly recovered.

\begin{figure}[t]
\begin{minipage}{0.31\linewidth}
\centering
\!\!\!\!\!\!\!\!\!\!\!\!\includegraphics[width=1.85in]{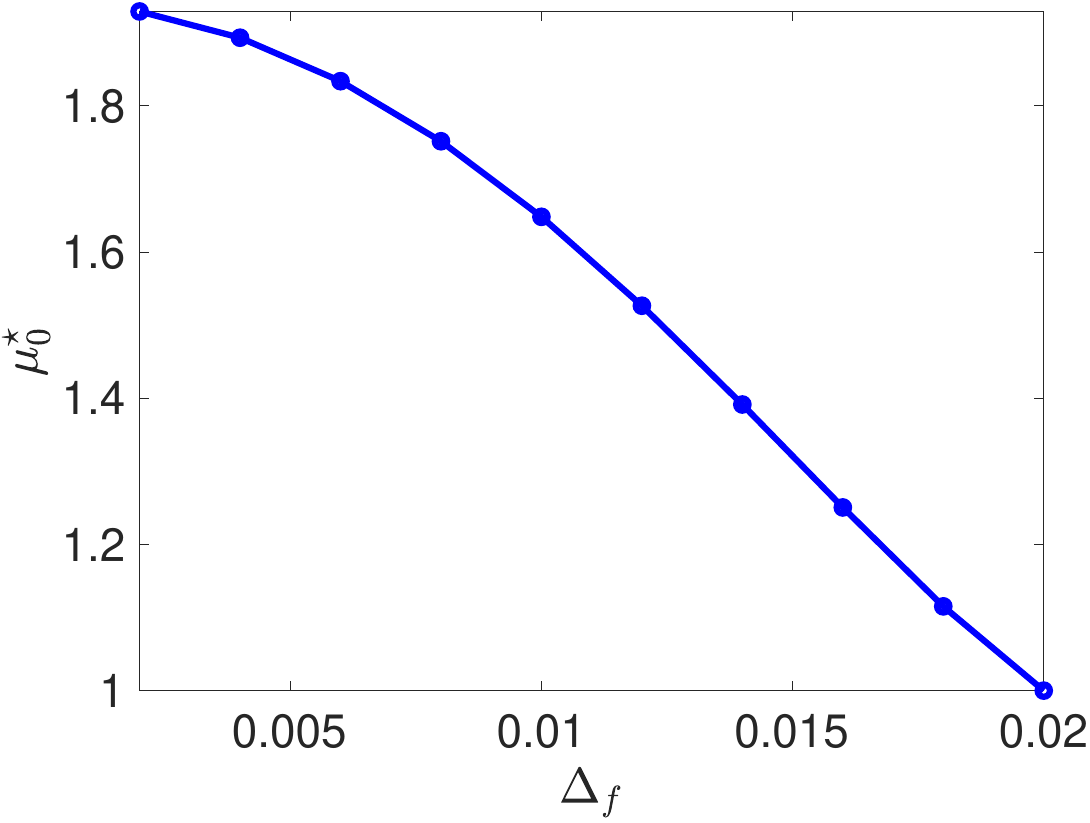}
\centerline{\!\!\!\!\!\!\!\!\!\!\!\!\footnotesize{(a)}}
\end{minipage}
~
\begin{minipage}{0.31\linewidth}
\centering
\!\!\!\!\!\!\!\!\!\includegraphics[width=1.85in]{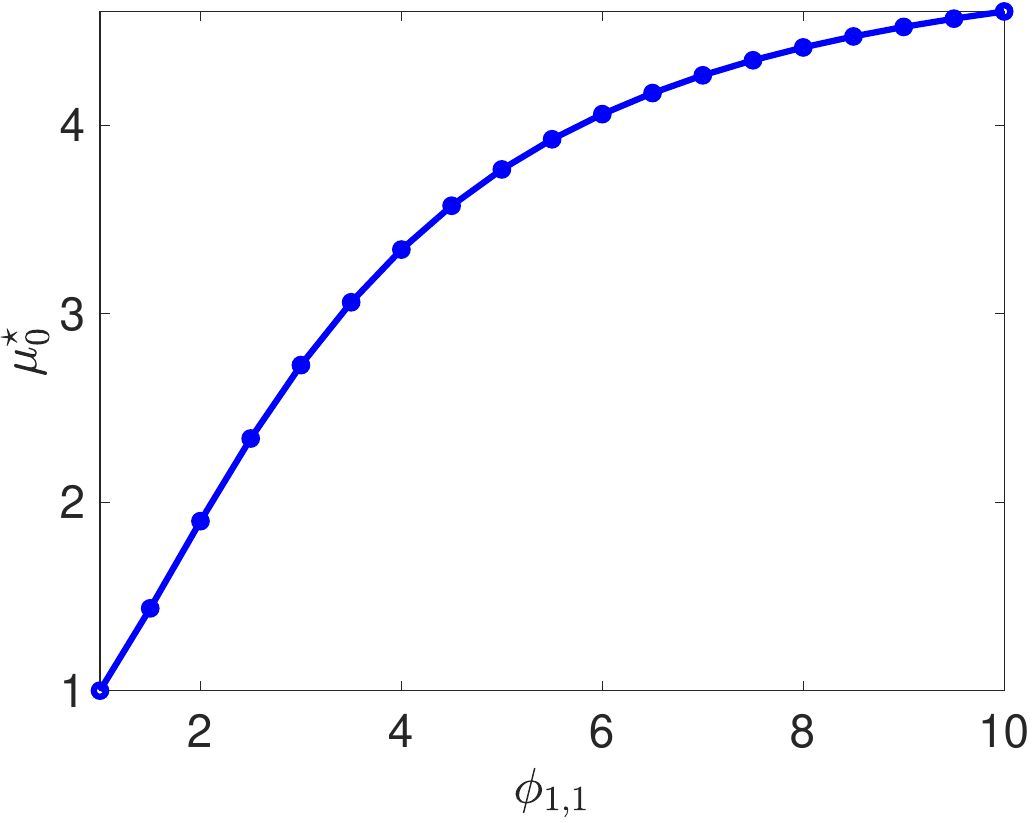}
\centerline{\!\!\!\!\!\!\!\!\!\footnotesize{(c)}}
\end{minipage}
~
\begin{minipage}{0.31\linewidth}
\centering
\!\!\!\includegraphics[width=1.85in]{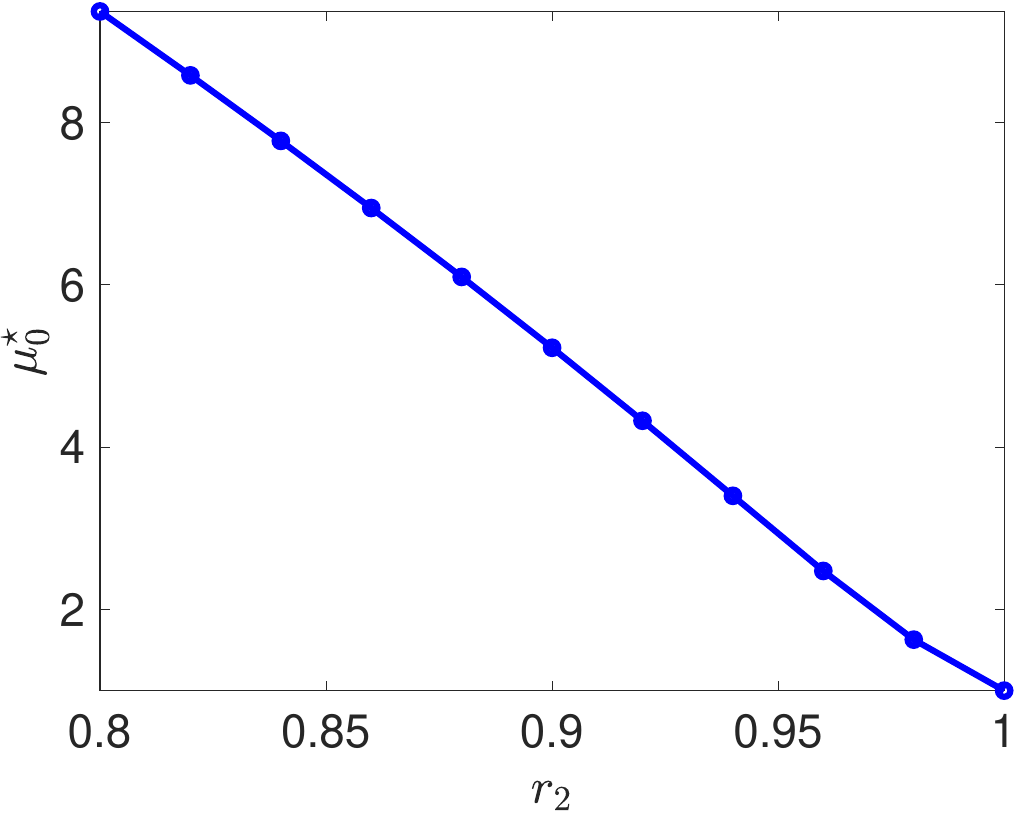}
\centerline{\!\!\!\footnotesize{(e)}}
\end{minipage}
\\
\begin{minipage}{0.32\linewidth}
\centering
\includegraphics[width=1.85in]{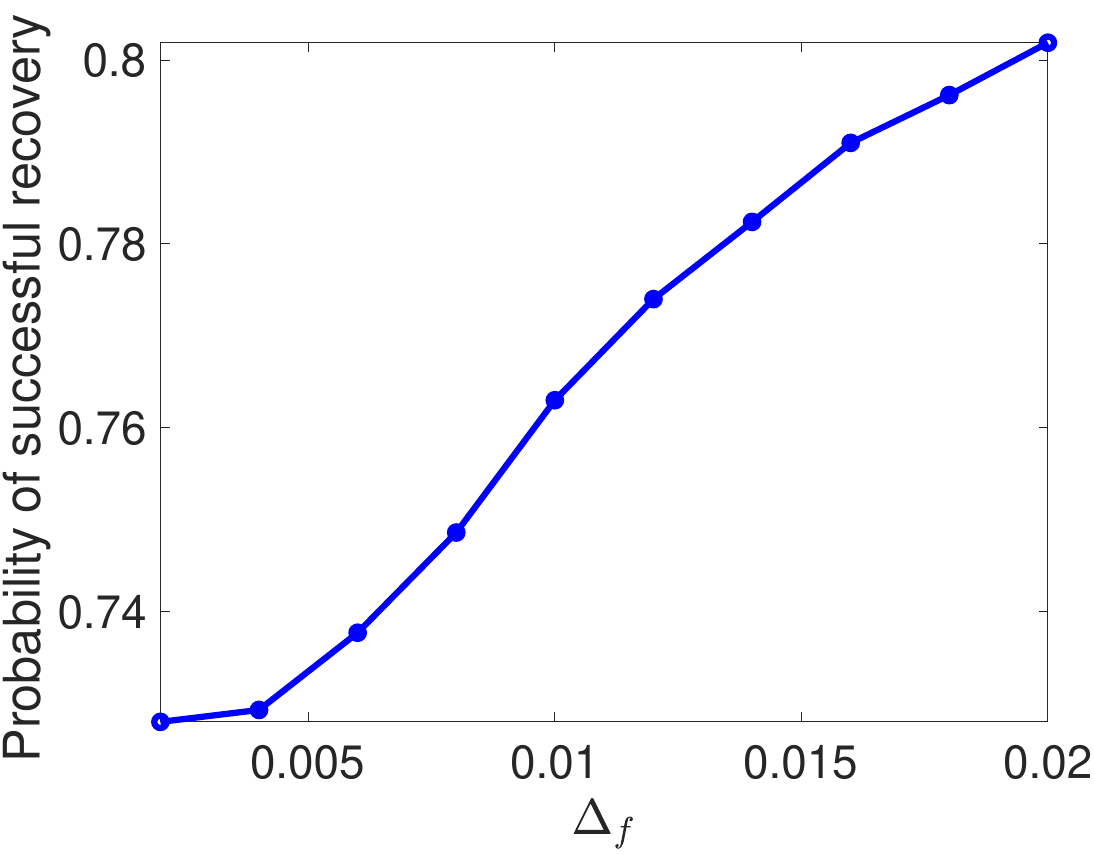}
\centerline{\footnotesize{(b)}}
\end{minipage}
~
\begin{minipage}{0.32\linewidth}
\centering
\includegraphics[width=1.85in]{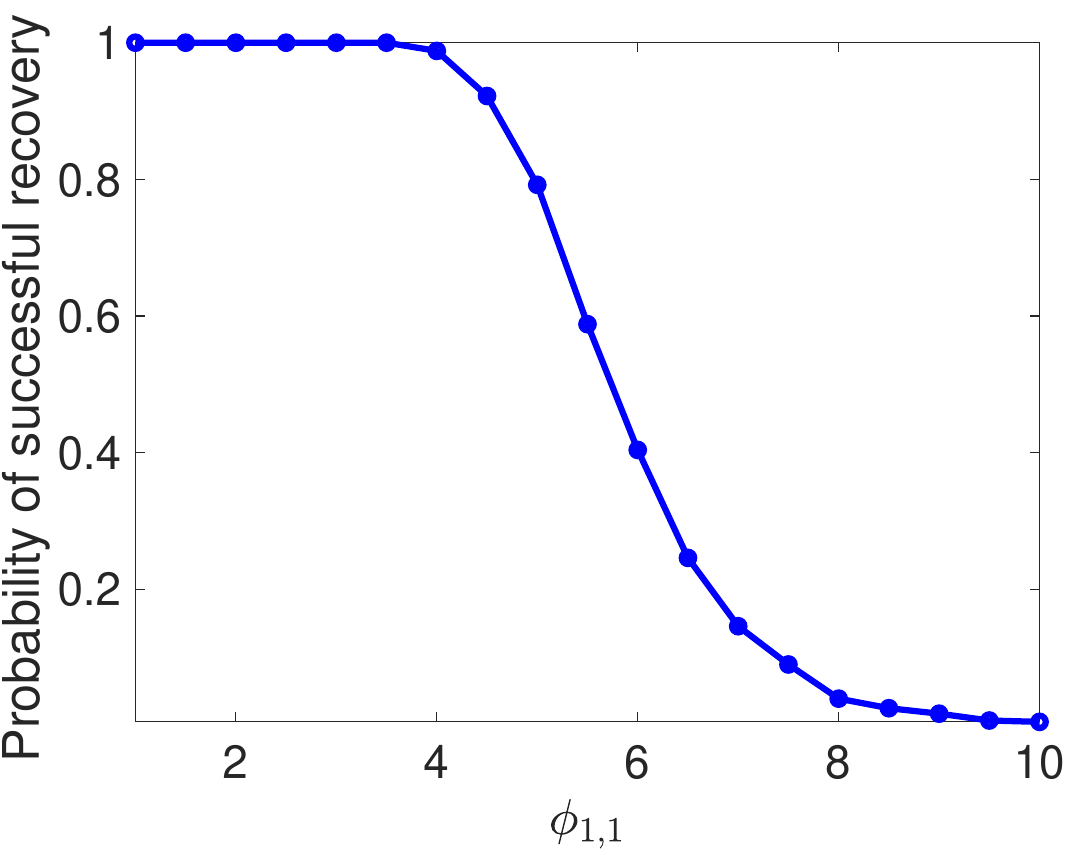}
\centerline{\footnotesize{(d)}}
\end{minipage}
~
\begin{minipage}{0.32\linewidth}
\centering
\includegraphics[width=1.85in]{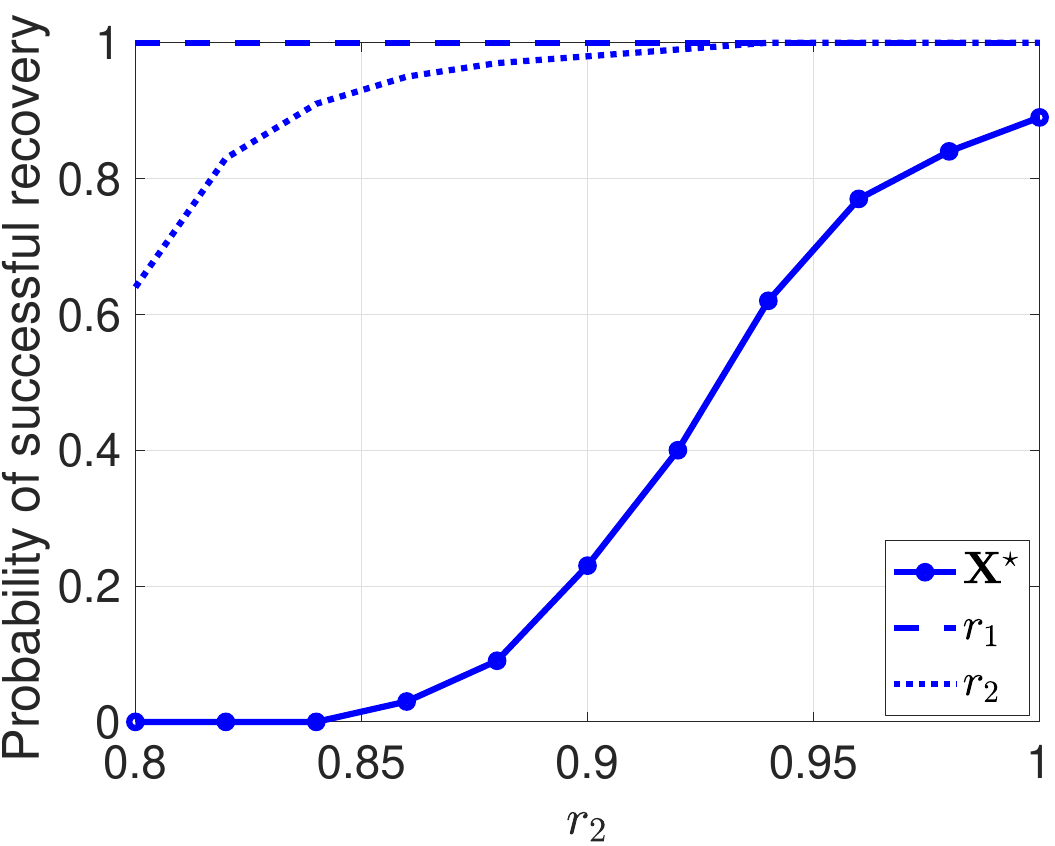}
\centerline{\footnotesize{(f)}}
\end{minipage}
\caption{Influence of minimum frequency separation $\Delta_f$, the matrix $\bPhi$, and damping ratio $r_2$ on the performance of missing data recovery: (a) coherence parameter of the data matrix $\Xs$ and (b) probability of successful data matrix recovery with respect to the minimum frequency separation $\Delta_f$. (c) coherence parameter of the data matrix $\Xs$ and (d) probability of successful data matrix recovery with respect to a various of matrix $\bPhi$. (e) coherence parameter of the data matrix $\Xs$ and (f) probability of successful data matrix and damping ratio recovery with respect to the second damping ratio $r_2$.}
\label{test_MUSIC_Missing_co_pa_deltf}
\end{figure}

These numerical experiments give a sense of how spectral parameters influence the coherence, and thus, recoverability of the data matrix. We stress again, however, that the significance of Corollary~\ref{THMdualploymiss} is that the sample complexity is not stated in terms of the matrix coherence (which may be difficult to immediately relate to the more tangible signal parameters); rather, the dependence on the damping ratios and minimum frequency separation is explicitly revealed in Corollary~\ref{THMdualploymiss}.

\subsection{Comparison with existing algorithms}

In this section, we implement a series of experiments to compare our proposed algorithms with three existing methods: 1) NNM$+$MUSIC\textcolor{black}{/ESPRIT}, 2) MN-MUSIC, and 3) ANM. We define successful parameter recovery as in~\eqref{succ_para}.

\subsubsection{NNM$+$MUSIC \textcolor{black}{and NNM$+$ESPRIT}}

We use NNM$+$MUSIC to denote an alternative approach wherein one first solves the NNM problem in~(\ref{SDPmiss}) to get $\Xh$ and then uses Algorithm~\ref{alg_NN-MUSIC} (or, equivalently, MUSIC) to identify the $r_k$'s and $f_k$'s from $\Xh$.\footnote{A similar idea has also been considered in~\cite{xu2018sep}.}
\textcolor{black}{We use NNM$+$ESPRIT to denote a similar two-stage approach but with MUSIC replaced by the Estimation of Signal Parameters via Rotation Invariance Techniques (ESPRIT) algorithm. The true $r_k$'s and $f_k$'s are set same as in Section~\ref{fdc}. We set $M=50$ and $N=20$.}
To show the advantage of our MD-MUSIC over NNM$+$MUSIC \textcolor{black}{and NNM$+$ESPRIT}, we present the probability of successful parameter recovery (defined in~\eqref{succ_para}) in Table~\ref{table1}.
In the ``two-step" algorithms NNM$+$MUSIC \textcolor{black}{and NNM$+$ESPRIT}, 
we use the true $K$ as the number of frequencies when we implement the MUSIC \textcolor{black}{or ESPRIT} algorithm \textcolor{black}{even though it is unknown and needs to estimated in practice}.

\begin{table}[ht]
\caption{\textcolor{black}{Comparison of the ``one-step" MD-MUSIC algorithm and the ``two-step" NNM$+$MUSIC and NNM$+$ESPRIT algorithms for parameter recovery.}
We present the probability of successful recovery over 1000 trials.} \label{table1}
\begin{center}
\begin{tabular}{|c|c|c|c|c|}
\hline
& \!$10\%$ missing \! & \!$20\%$ missing\!  & \!$30\%$ missing\!      &  \!$40\%$ missing \! \\
\hline
MD-MUSIC &$99.6\%$&$96.9\%$& $90.8\%$&$78.8\%$\\
\hline
      NNM$+$MUSIC      & $99.4\%$&$96.3\%$& $85.0\%$&$56.2\%$\\
\hline
      \textcolor{black}{NNM$+$ESPRIT}      & $98.8\%$&$89.9\%$& $66.6\%$&$29.5\%$\\
\hline
\end{tabular}
\end{center}
\end{table}

\subsubsection{MN-MUSIC}

Next, we compare our proposed MD-MUSIC algorithm with the MN-MUSIC algorithm introduced in Section~\ref{MN-MUSIC}~\cite{suryaprakash2012performance} in a scenario where $20\%$ of the noiseless data entries are missing. We observe that the MN-MUSIC algorithm never successfully recovers the frequencies and damping ratios
since it performs an SVD directly on the missing data.\footnote{No results are shown since MN-MUSIC never recovers successfully.}

\subsubsection{ANM}

Finally, we compare the proposed NN-MUSIC and MD-MUSIC algorithms with ANM in the full and missing data cases, respectively. In ANM, we solve the following SDP
\begin{align*}
\min_{\X,\u,\D}~\frac{1}{2M} \change{Tr}(\change{Toep}(\u))+\frac{1}{2} \change{Tr}(\D) \quad 
\change{s.t.}~\left[
\begin{array}{cc}
\change{Toep}(\u)&\X\\
\X^H&\D
\end{array}
\right]\succeq 0, ~\X_{\Omega} = \Xs_{\Omega},
\end{align*}
where $\change{Toep}(\u)$ is a Hermitian Toeplitz matrix with the vector $\u$
being its first column. We use $\X= \Xs$ instead of $\X_{\Omega} = \Xs_{\Omega}$ in the full data case.
Similar to NN-MUSIC and MD-MUSIC, given the dual solution of the above SDP, we then formulate a dual polynomial and localize the places where the $\ell_2$-norm of the dual polynomial achieves~$1$ to extract the estimated frequencies. Since ANM can only recover frequencies, we only compare the accuracy of estimated frequencies in this section.
All the simulation results presented in this section are an average over 100 trials.

In the full data case, we repeat the first experiment in Section~\ref{fdc} with $N = 10$ and with a variety of $M$ and $\Delta_f$. We firstly fix $\Delta_f = 0.06$ and set the true frequency and damping pairs as $(r_1,f_1) = (0.86, 0.1)$, $(r_2,f_2) = (0.92, 0.16)$, and $(r_3,f_3) = (0.98, 0.8)$. Then, we compare NN-MUSIC and ANM with a variety of $M$. Next, we fix $M=20$, $r_1= 0.92$, $r_2 = 0.98$, and $f_1 = 0.1$. Similar as in Section~\ref{test_coherence}, we then let $f_2 = f_1+\Delta_f$ with various values of $\Delta_f$. The simulation results are given in Figure~\ref{test_NNvsANM}. It can be seen that the NN-MUSIC algorithm significantly outperforms ANM and can always recover the frequencies exactly, as indicated in Theorem~\ref{THMdualploy}. This is because our data contains damping, which is not modeled in ANM.

In the missing data case, we randomly remove $20\%$ or $40\%$ of the data entries. We repeat the above two experiments with these partially observed data matrices to compare MD-MUSIC and ANM. As shown in Figure~\ref{test_MDvsANM}, MD-MUSIC still outperforms ANM significantly in most cases due to its ability to handle damped signals. We also observe that ANM can have a higher probability of successful recovery once the number of observed entries is too small, as shown in Figure~\ref{test_MDvsANM} (c). However, the success probability in this case is still significantly less than~$1$. Note that we have changed $f_2$ from 0.16 to 0.2 in Figure~\ref{test_MDvsANM} (a, c) to test with a larger value of $\Delta_f$. Other parameters used in this part are the same as in the full data experiments.

\begin{figure}[t]
\begin{minipage}{0.23\linewidth}
	
\end{minipage}
\hfill
\begin{minipage}{0.23\linewidth}
\centering
\includegraphics[width=1.3in]{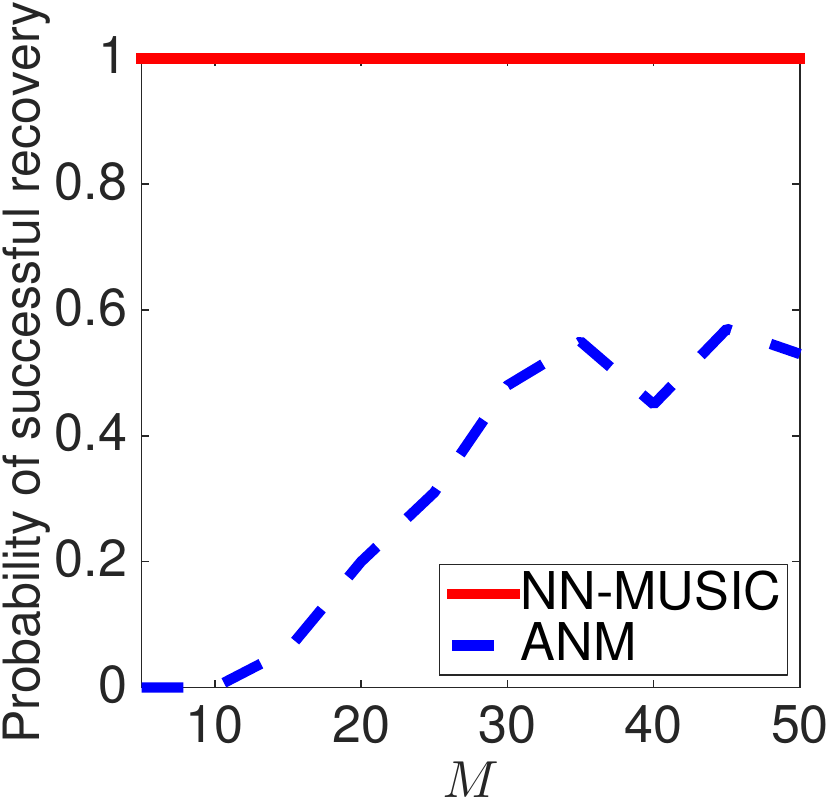}
\centerline{\footnotesize{(a)}}
\end{minipage}
\hfill
\begin{minipage}{0.23\linewidth}
\centering
\includegraphics[width=1.3in]{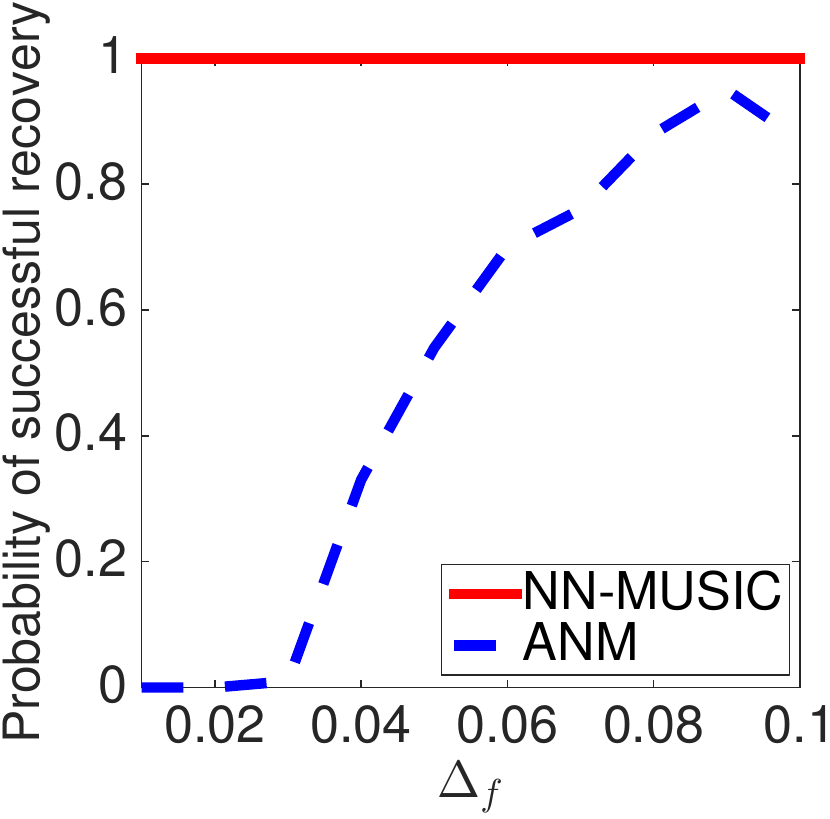}
\centerline{\footnotesize{(b)}}
\end{minipage}
\hfill
\begin{minipage}{0.23\linewidth}
	
\end{minipage}
\caption{Comparison of NN-MUSIC and ANM in the noiseless full data case, with damped exponentials. (a) Probability of successful frequency recovery as a function of $M$, with fixed $\Delta_f = 0.06$ and $K=3$. (b) Probability of successful frequency recovery as a function of $\Delta_f$, with fixed $M=20$ and $K=2$. }
\label{test_NNvsANM}
\end{figure}

\begin{figure}[t]
\begin{minipage}{0.23\linewidth}
\centering
\includegraphics[width=1.3in]{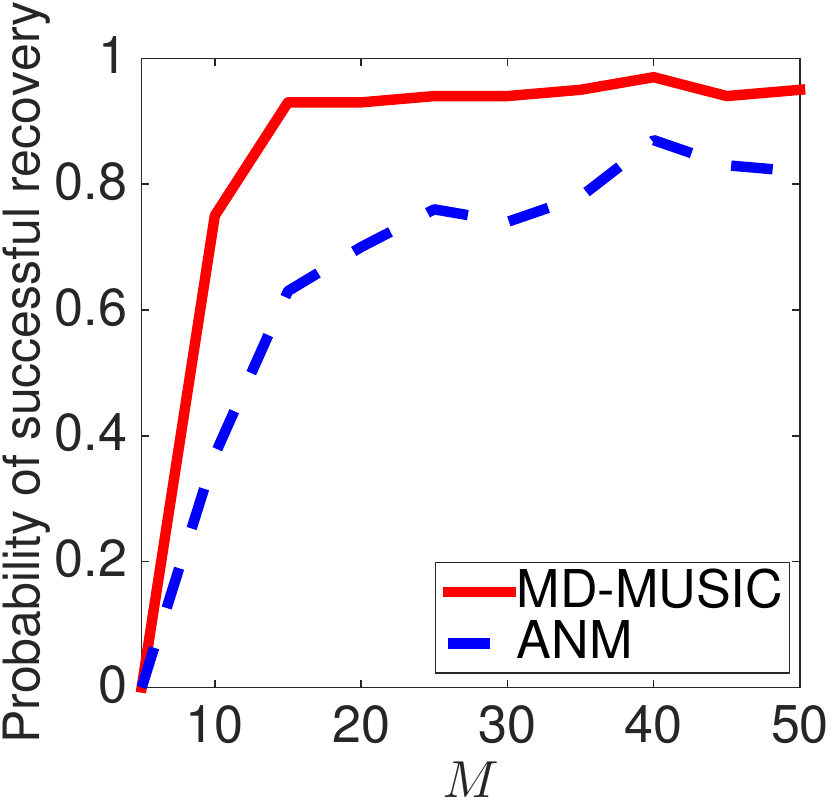}
\centerline{\footnotesize{(a) $20\%$ missing}}
\end{minipage}
\hfill
\begin{minipage}{0.23\linewidth}
\centering
\includegraphics[width=1.3in]{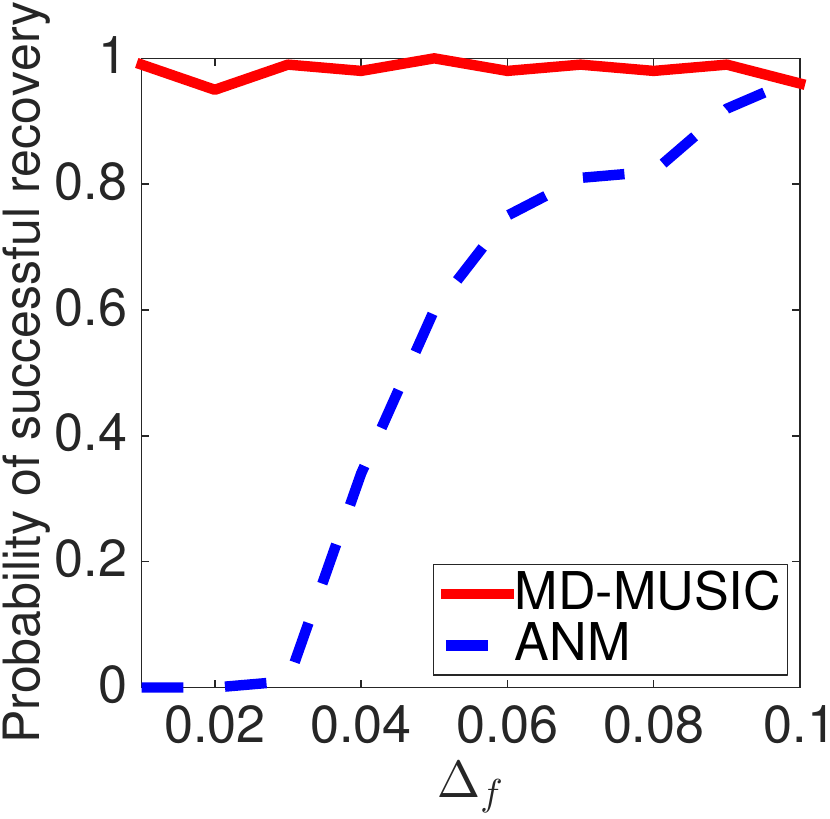}
\centerline{\footnotesize{(b) $20\%$ missing}}
\end{minipage}
\hfill
\begin{minipage}{0.23\linewidth}
\centering
\includegraphics[width=1.3in]{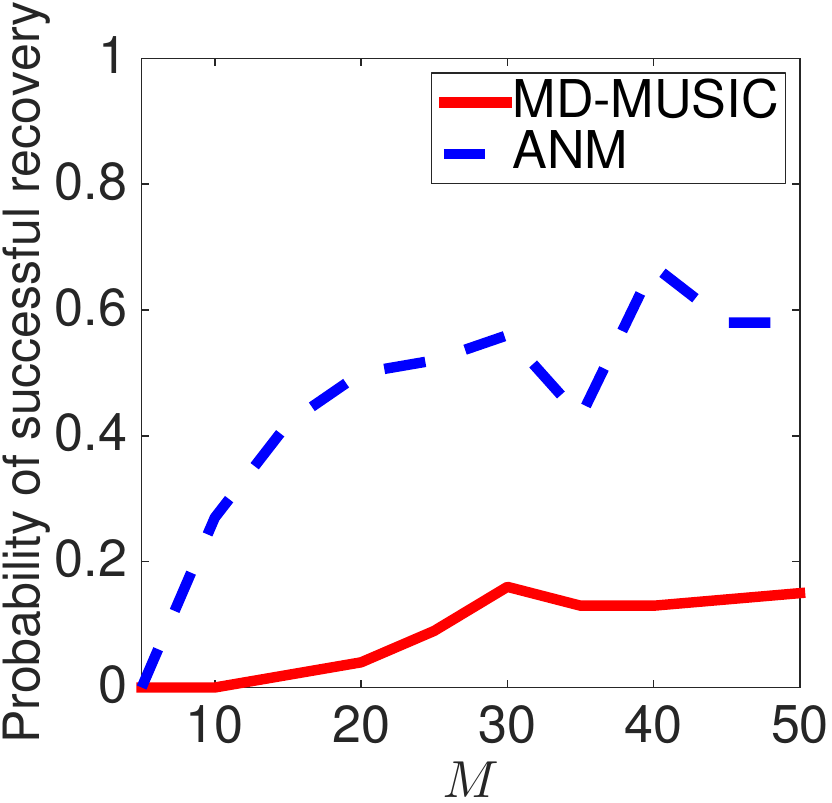}
\centerline{\footnotesize{(c) $40\%$ missing}}
\end{minipage}
\hfill
\begin{minipage}{0.23\linewidth}
\centering
\includegraphics[width=1.3in]{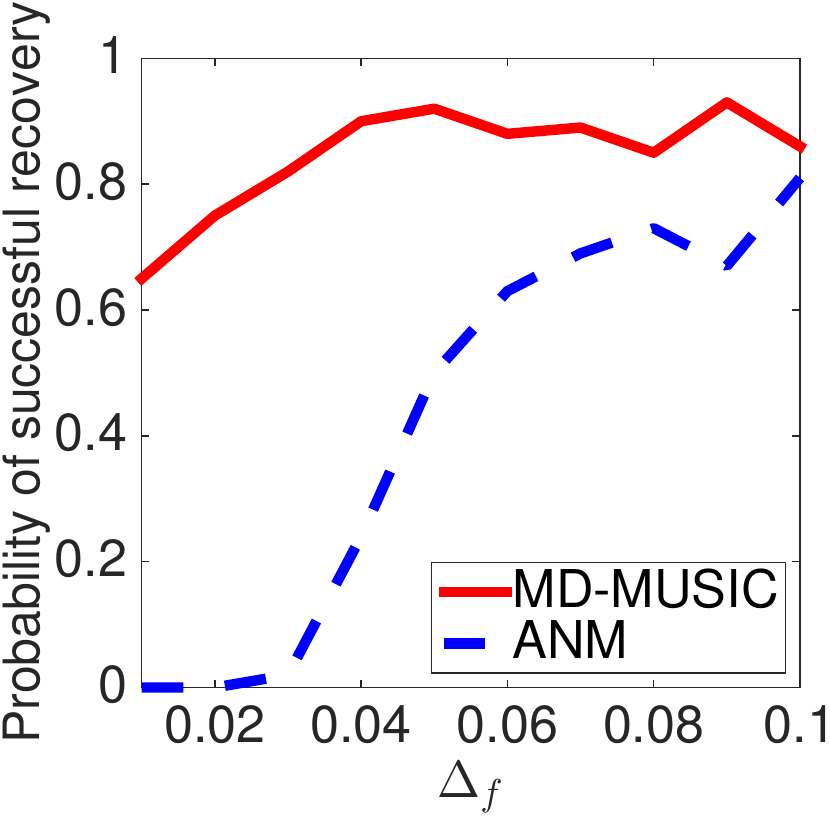}
\centerline{\footnotesize{(d) $40\%$ missing}}
\end{minipage}
\caption{Comparison of MD-MUSIC and ANM in the noiseless missing data case and damped exponentials, with (a, b) $20\%$ or (c, d) $40\%$ of the entries randomly removed. (a), (c) Probability of successful frequency recovery as a function of $M$, with fixed $\Delta_f = 0.1$ and $K=3$. (b), (d) Probability of successful frequency recovery as a function of $\Delta_f$, with fixed $M=20$ and $K=2$.}
\label{test_MDvsANM}
\end{figure}

\section{Conclusion}
\label{conc}

In this work, we provide a convex optimization view for the classical MUSIC algorithm in spectral estimation with damping. In particular, we build a connection between NNM and the classical MUSIC algorithm, which inspires us to propose a new algorithm, named MD-MUSIC, for the missing data field. Theoretical results are provided to guarantee the proposed algorithms.   \textcolor{black}{In particular, it is possible to get exact parameter recovery with the MD-MUSIC algorithm even when we do not have perfect data recovery. Moreover, for the missing data case, we also quantify how the sample complexity depends on the true spectral parameters rather than use certain incoherence properties as in existing literature.} Meanwhile, numerical simulations indicate that the proposed algorithms work very well and significantly outperform some relevant existing methods in frequency estimation of damped exponentials. We leave the robust performance analysis on noisy data for future work.

\section*{Acknowledgement}

MW and SL were supported by NSF grant CCF--1409258, NSF CAREER grant CCF--1149225, and NSF grant CCF-1704204.

\bibliographystyle{ieeetr}
\bibliography{sample}
%


%

\end{document}